\documentclass[reqno,11pt,a4paper]{amsart}
\usepackage{mathrsfs}
\usepackage{amssymb}
\usepackage{amsfonts}
\usepackage{latexsym}
\usepackage{amsthm}

\usepackage{graphicx}

\usepackage{graphicx,xcolor}
\def\lb{\label}

\newcommand{\er}[1]{\textrm{(\ref{#1})}}

\begin{document}


\renewcommand{\theequation}{\arabic{section}.\arabic{equation}}
\theoremstyle{plain}
\newtheorem{theorem}{\bf Theorem}[section]
\newtheorem{lemma}[theorem]{\bf Lemma}
\newtheorem{corollary}[theorem]{\bf Corollary}
\newtheorem{proposition}[theorem]{\bf Proposition}
\newtheorem{definition}[theorem]{\bf Definition}
\newtheorem{remark}[theorem]{\it Remark}
\newtheorem{example}[theorem]{\it Example}


\def\a{\alpha}  \def\cA{{\mathcal A}}     \def\bA{{\bf A}}  \def\mA{{\mathscr A}}
\def\b{\beta}   \def\cB{{\mathcal B}}     \def\bB{{\bf B}}  \def\mB{{\mathscr B}}
\def\g{\gamma}  \def\cC{{\mathcal C}}     \def\bC{{\bf C}}  \def\mC{{\mathscr C}}
\def\G{\Gamma}  \def\cD{{\mathcal D}}     \def\bD{{\bf D}}  \def\mD{{\mathscr D}}
\def\d{\delta}  \def\cE{{\mathcal E}}     \def\bE{{\bf E}}  \def\mE{{\mathscr E}}
\def\D{\Delta}  \def\cF{{\mathcal F}}     \def\bF{{\bf F}}  \def\mF{{\mathscr F}}
\def\c{\chi}    \def\cG{{\mathcal G}}     \def\bG{{\bf G}}  \def\mG{{\mathscr G}}
\def\z{\zeta}   \def\cH{{\mathcal H}}     \def\bH{{\bf H}}  \def\mH{{\mathscr H}}
\def\e{\eta}    \def\cI{{\mathcal I}}     \def\bI{{\bf I}}  \def\mI{{\mathscr I}}
\def\p{\psi}    \def\cJ{{\mathcal J}}     \def\bJ{{\bf J}}  \def\mJ{{\mathscr J}}
\def\vT{\Theta} \def\cK{{\mathcal K}}     \def\bK{{\bf K}}  \def\mK{{\mathscr K}}
\def\k{\kappa}  \def\cL{{\mathcal L}}     \def\bL{{\bf L}}  \def\mL{{\mathscr L}}
\def\l{\lambda} \def\cM{{\mathcal M}}     \def\bM{{\bf M}}  \def\mM{{\mathscr M}}
\def\L{\Lambda} \def\cN{{\mathcal N}}     \def\bN{{\bf N}}  \def\mN{{\mathscr N}}
\def\m{\mu}     \def\cO{{\mathcal O}}     \def\bO{{\bf O}}  \def\mO{{\mathscr O}}
\def\n{\nu}     \def\cP{{\mathcal P}}     \def\bP{{\bf P}}  \def\mP{{\mathscr P}}
\def\r{\rho}    \def\cQ{{\mathcal Q}}     \def\bQ{{\bf Q}}  \def\mQ{{\mathscr Q}}
\def\s{\sigma}  \def\cR{{\mathcal R}}     \def\bR{{\bf R}}  \def\mR{{\mathscr R}}
\def\S{\Sigma}  \def\cS{{\mathcal S}}     \def\bS{{\bf S}}  \def\mS{{\mathscr S}}
\def\t{\tau}    \def\cT{{\mathcal T}}     \def\bT{{\bf T}}  \def\mT{{\mathscr T}}
\def\f{\phi}    \def\cU{{\mathcal U}}     \def\bU{{\bf U}}  \def\mU{{\mathscr U}}
\def\F{\Phi}    \def\cV{{\mathcal V}}     \def\bV{{\bf V}}  \def\mV{{\mathscr V}}
\def\P{\Psi}    \def\cW{{\mathcal W}}     \def\bW{{\bf W}}  \def\mW{{\mathscr W}}
\def\o{\omega}  \def\cX{{\mathcal X}}     \def\bX{{\bf X}}  \def\mX{{\mathscr X}}
\def\x{\xi}     \def\cY{{\mathcal Y}}     \def\bY{{\bf Y}}  \def\mY{{\mathscr Y}}
\def\X{\Xi}     \def\cZ{{\mathcal Z}}     \def\bZ{{\bf Z}}  \def\mZ{{\mathscr Z}}
\def\O{\Omega}
\def\vr{\varrho}
\def\vs{\varsigma}

\newcommand{\gA}{\mathfrak{A}}          \newcommand{\ga}{\mathfrak{a}}
\newcommand{\gB}{\mathfrak{B}}          \newcommand{\gb}{\mathfrak{b}}
\newcommand{\gC}{\mathfrak{C}}          \newcommand{\gc}{\mathfrak{c}}
\newcommand{\gD}{\mathfrak{D}}          \newcommand{\gd}{\mathfrak{d}}
\newcommand{\gE}{\mathfrak{E}}
\newcommand{\gF}{\mathfrak{F}}           \newcommand{\gf}{\mathfrak{f}}
\newcommand{\gG}{\mathfrak{G}}           
\newcommand{\gH}{\mathfrak{H}}           \newcommand{\gh}{\mathfrak{h}}
\newcommand{\gI}{\mathfrak{I}}           \newcommand{\gi}{\mathfrak{i}}
\newcommand{\gJ}{\mathfrak{J}}           \newcommand{\gj}{\mathfrak{j}}
\newcommand{\gK}{\mathfrak{K}}            \newcommand{\gk}{\mathfrak{k}}
\newcommand{\gL}{\mathfrak{L}}            \newcommand{\gl}{\mathfrak{l}}
\newcommand{\gM}{\mathfrak{M}}            \newcommand{\gm}{\mathfrak{m}}
\newcommand{\gN}{\mathfrak{N}}            \newcommand{\gn}{\mathfrak{n}}
\newcommand{\gO}{\mathfrak{O}}
\newcommand{\gP}{\mathfrak{P}}             \newcommand{\gp}{\mathfrak{p}}
\newcommand{\gQ}{\mathfrak{Q}}             \newcommand{\gq}{\mathfrak{q}}
\newcommand{\gR}{\mathfrak{R}}             \newcommand{\gr}{\mathfrak{r}}
\newcommand{\gS}{\mathfrak{S}}              \newcommand{\gs}{\mathfrak{s}}
\newcommand{\gT}{\mathfrak{T}}             \newcommand{\gt}{\mathfrak{t}}
\newcommand{\gU}{\mathfrak{U}}             \newcommand{\gu}{\mathfrak{u}}
\newcommand{\gV}{\mathfrak{V}}             \newcommand{\gv}{\mathfrak{v}}
\newcommand{\gW}{\mathfrak{W}}             \newcommand{\gw}{\mathfrak{w}}
\newcommand{\gX}{\mathfrak{X}}               \newcommand{\gx}{\mathfrak{x}}
\newcommand{\gY}{\mathfrak{Y}}              \newcommand{\gy}{\mathfrak{y}}
\newcommand{\gZ}{\mathfrak{Z}}             \newcommand{\gz}{\mathfrak{z}}

\def\ba{{\bf a}}\def\be{{\bf e}} \def\bc{{\bf c}}
\def\bv{{\bf v}} \def\bu{{\bf u}}

\def\be{{\bf e}} \def\bc{{\bf c}}
\def\bx{{\bf x}} \def\by{{\bf y}}
\def\bv{{\bf v}} \def\bu{{\bf u}}

\def\Om{\Omega}
\def\bbD{\pmb \Delta}
\def\mm{\mathrm m}
\def\mn{\mathrm n}

\def\ve{\varepsilon}   \def\vt{\vartheta}    \def\vp{\varphi}
   \def\vk{\varkappa}

 \def\dD{{\mathbb A}}   \def\B{{\mathbb B}} \def\C{{\mathbb C}}
 \def\dD{{\mathbb D}}  \def\dE{{\mathbb E}}  \def\dG{{\mathbb G}}
 \def\dF{{\mathbb F}}  \def\dI{{\mathbb I}}  \def\dJ{{\mathbb J}}
 \def\K{{\mathbb K}}  \def\dL{{\mathbb L}}   \def\dM{{\mathbb M}}
 \def\N{{\mathbb N}}  \def\dO{{\mathbb O}}  \def\dP{{\mathbb P}}
 \def\R{{\mathbb R}}  \def\dS{{\mathbb S}}  \def\T{{\mathbb T}}
  \def\dU{{\mathbb U}}  \def\dV{{\mathbb V}} \def\dW{{\mathbb W}}
   \def\dX{{\mathbb X}} \def\dY{{\mathbb Y}} \def\Z{{\mathbb Z}}


\def\la{\leftarrow}              \def\ra{\rightarrow}      \def\Ra{\Rightarrow}
\def\ua{\uparrow}                \def\da{\downarrow}
\def\lra{\leftrightarrow}        \def\Lra{\Leftrightarrow}


\def\lt{\biggl}                  \def\rt{\biggr}
\def\ol{\overline}               \def\wt{\widetilde}
\def\no{\noindent}               \def\ti{\tilde}
\def\ul{\underline}


\let\ge\geqslant                 \let\le\leqslant
\def\lan{\langle}                \def\ran{\rangle}
\def\/{\over}                    \def\iy{\infty}
\def\sm{\setminus}               \def\es{\emptyset}
\def\ss{\subset}                 \def\ts{\times}
\def\pa{\partial}                \def\os{\oplus}
\def\om{\ominus}                 \def\ev{\equiv}
\def\iint{\int\!\!\!\int}        \def\iintt{\mathop{\int\!\!\int\!\!\dots\!\!\int}\limits}
\def\el2{\ell^{\,2}}             \def\1{1\!\!1}
\def\wh{\widehat}

\def\sh{\mathop{\mathrm{sh}}\nolimits}
\def\ch{\mathop{\mathrm{ch}}\nolimits}

\def\where{\mathop{\mathrm{where}}\nolimits}
\def\as{\mathop{\mathrm{as}}\nolimits}
\def\Area{\mathop{\mathrm{Area}}\nolimits}
\def\arg{\mathop{\mathrm{arg}}\nolimits}
\def\const{\mathop{\mathrm{const}}\nolimits}
\def\det{\mathop{\mathrm{det}}\nolimits}
\def\diag{\mathop{\mathrm{diag}}\nolimits}
\def\diam{\mathop{\mathrm{diam}}\nolimits}
\def\dim{\mathop{\mathrm{dim}}\nolimits}
\def\dist{\mathop{\mathrm{dist}}\nolimits}
\def\Im{\mathop{\mathrm{Im}}\nolimits}
\def\Iso{\mathop{\mathrm{Iso}}\nolimits}
\def\Ker{\mathop{\mathrm{Ker}}\nolimits}
\def\Lip{\mathop{\mathrm{Lip}}\nolimits}
\def\rank{\mathop{\mathrm{rank}}\limits}
\def\Ran{\mathop{\mathrm{Ran}}\nolimits}
\def\Re{\mathop{\mathrm{Re}}\nolimits}
\def\Res{\mathop{\mathrm{Res}}\nolimits}
\def\res{\mathop{\mathrm{res}}\limits}
\def\sign{\mathop{\mathrm{sign}}\nolimits}
\def\span{\mathop{\mathrm{span}}\nolimits}
\def\supp{\mathop{\mathrm{supp}}\nolimits}
\def\Tr{\mathop{\mathrm{Tr}}\nolimits}
\def\BBox{\hspace{1mm}\vrule height6pt width5.5pt depth0pt \hspace{6pt}}


\newcommand\nh[2]{\widehat{#1}\vphantom{#1}^{(#2)}}
\def\dia{\diamond}

\def\Oplus{\bigoplus\nolimits}



\def\qqq{\qquad}
\def\qq{\quad}
\let\ge\geqslant
\let\le\leqslant
\let\geq\geqslant
\let\leq\leqslant
\newcommand{\ca}{\begin{cases}}
\newcommand{\ac}{\end{cases}}
\newcommand{\ma}{\begin{pmatrix}}
\newcommand{\am}{\end{pmatrix}}
\renewcommand{\[}{\begin{equation}}
\renewcommand{\]}{\end{equation}}
\def\bu{\bullet}
\def\tes{\textstyle}

\baselineskip14pt

\title[{Scattering for time periodic Hamiltonians on graphs  }]
 {Scattering for time periodic Hamiltonians on graphs}

\author[Hiroshi Isozaki]{Hiroshi Isozaki}
\address{Graduate School of Pure and Applied Sciences,
University of Tsukuba, 1-1-1, Tennoudai, Tsukuba, 305-8571, Japan,
isozakih@math.tsukuba.ac.jp}

\author[Evgeny, L. Korotyaev]{Evgeny, L. Korotyaev}
\address{E. Korotyaev, Department of Mathematical Analysis, Saint-Petersburg State University, Universitetskaya nab. 7/9, St. Petersburg,
 korotyaev@gmail.com}

\date{\today}

\begin{abstract}
\no  We  develop a scattering theory for time-periodic Hamiltonians on
discrete graphs, including long-range potentials with zero average for
the period, and show the existence and completeness of wave operators.
\end{abstract}

\subjclass{34A55, (34B24, 47E05)} \keywords{trace formula,
time-periodic potentials}

\maketitle

\baselineskip 15pt

\section {Introduction}
\subsection{Space-time periodic systems}
We consider a scattering problem for a magnetic Schr\"odinger equation on a
periodic graph $\cG=(\cV,\cE)$:
\[
\lb{eg1xx}
\begin{aligned}
{\frac{d}{dt}}u(t)= -i h(t) u(t),\qq h(t)=\D_{\b(t)}+\gp+ V(t),
\end{aligned}
\]
where the Hamiltonian $h(t) $ is  $\t$-periodic in continuous time
$t$ in the following sense: $\D_{\b(t)}$ is a discrete magnetic
Laplacian with  magnetic potential  $\b$ on $\cV$ which is
$\t$-periodic in time (see \er{1}). The electric potential $V(t)$ on
$\cV$ is also assumed to be $\t$-periodic in time, and $\gp:\cV\to
\R $ is a time-independent bounded electric potential.

The above operator has a lot of applications to various periodic
media, e.g. nanomedia, in physics, chemistry and engineering, see,
e.g., \cite{NG04}.  Approximation by periodic graphs are often used
to study  properties of such media, and the problem is reduced to
Schr\"odinger operators on graphs. It is known that the spectrum of
Schr\"odinger operators with periodic potentials on periodic
discrete graphs consists of a finite number of flat bands, i.e.
eigenvalues  of infinite multiplicities, and an absolutely
continuous part (a union of a finite number of non-degenerate
bands), the latter of which gives rise to scattering. We are
interested in the scattering by periodic Hamiltonians on graphs,
which are also periodic in time.

\subsection{Magnetic Schr\"odinger operators on graphs}
Let $\cG=(\cV,\cE)$ be a connected infinite graph embedded in
$\R^d$. Here $\cV$ is the set of vertices and $\cE$   is the set of
non-oriented edges. For $x, y \in \cV$, if there is an edge in $\cE$
starting from $x$ and ending at $y$, we say that $x$ and $y$ are
adjacent and denote by $x\sim y$. Each edge is not oriented at
first, however, we give two orientations for the edges in $\cE$, and
denote  the set of all oriented edges by $\cA$.  The edge starting
from $x$  and ending at  $y$ is denoted by the ordered pair
$(x,y)\in\cA$, and is said to be \emph{incident} to the vertices. We
define the degree ${\vk}_x\ge 1$ of a vertex $x\in\cV$ as the number
of all edges in $\cA$ starting from $x$. We assume that
\begin{equation}
\vk_+:=\sup_{x\in\cV}\vk_x<\iy.
\label{kappaxdefine}
\end{equation}
A sequence of directed edges $(\be_1,\be_2,\ldots,\be_n)$ is called
a \emph{path} if the terminus of the edge $\be_s$ coincides with the
origin of the edge $\be_{s+1}$ for all $s=1,\ldots,n-1$. If the
terminus of $\be_n$ coincides with the origin of $\be_1$, then the
path is called a \emph{cycle}. The inverse edge of $\be=(x,y)\in\cA$
is denoted by $\ul \be=(y,x)$. A \emph{magnetic vector potential on
$\cG$} is  a function $\a:\cA\ra\R$ satisfying
$\a(\ul\be\,)=-\a(\be)$ for all $\be\in \cA$. Let $\ell^p(\cV),
p\ge1$ be the space of all functions $f:\cV\to \C$ equipped with the
norm $ \|f\|^p_{\ell^p(\cV)}=\sum_{x\in\cV}|f_x|^p<\infty. $ For a
magnetic vector potential $\a:\cA\ra\R$,  we define the
\emph{combinatorial magnetic Laplacian} $\D_\a$ on $f=(f_x)_{x\in
\cV}\in\ell^2(\cV)$ by
\[
\lb{1} \big(\D_{\a}
f\big)(x)={1\/2}\sum_{\be=(x,y)\in\cA}\big(f_x-e^{i\a(\be)}f_y\big),
\qqq x\in\cV,
\]
where the sum is taken over all oriented edges starting from
the vertex $x$. It is well-known (see, e.g., \cite{MW89})
that $\D_{\a}$is self-adjoint and its spectrum $\s(\D_{\a})$ satisfies
\[
\lb{2}
\begin{aligned}
\s(\D_{\a})\subset[0,\vk_+].
\end{aligned}
\]
If $\a=0$, then $\D=\D_0$ is the usual Laplacian $\D$ without
magnetic potential:
\begin{equation}
 \big(\D f\big)(x)={1\/2}\sum_{(x,y)\in\cA}\big(f_x-f_y\big),
\qqq x\in\cV.
\label{DefineDeltaforV}
\end{equation}
 Remark that there are other definitions of discrete magnetic
Laplacians on graphs: weighted,  normalized, standard Laplacians,
see e.g. \cite{MW89}, \cite{KS17} and references therein. For a
magnetic Laplacian $\D_\a$ and  a real potential $\gp \in
\ell^\iy(\cV)$, that is  $ \big(\gp f\big)(x)=\gp_xf_x$ for all
$x\in \cV$, we define a magnetic Schr\"odinger operator $h_\a$ with
an electric potential $\gp$ on the Hilbert space $\ell^2(\cV)$ by
\[
\lb{4} h_\a=\D_{\a}+\gp.
\]

\subsection{ Periodic graphs.} 
Let $\G$ be a lattice of rank $d$ in
$\R^d$ with a basis $\{\ba_1,\ldots,\ba_d\}$, i.e., $ \G=\Big\{\ba :
\ba=\sum_{j=1}^dn_j\ba_j, \; n_j\in\Z\Big\}$.
We define an equivalence relation on $\R^d$:
$$
x=y \; (\hspace{-4mm}\mod \G) \qq\Leftrightarrow\qq x-y\in\G \qqq
\forall\, x,y\in\R^d.
$$
We consider \emph{a locally finite $\G$-periodic graph} $\cG$, i.e. a
graph satisfying the following condition:

\medskip
$\bu $ $\cG=\cG+\ba$ for any $\ba\in\G$ and the quotient graph
$\cG_*=\cG/\G$ is finite.

\medskip
\noindent The basis $\ba_1,\ldots,\ba_d$ of the lattice $\G$ is
called the {\it periods}  of $\cG$. The \emph{fundamental graph}
$\cG_*=\cG/\G$ is  the quotient graph of $\cG$ by $\Gamma$, which is
a graph on the $d$-dimensional torus $\R^d/\G$. It has the vertex
set $\cV_*=\cV/\G$, the set $\cE_*=\cE/\G$ of unoriented edges and
the set $\cA_*=\cA/\G$ of oriented edges, all of which are finite.

A magnetic vector potential $\a:\cA\ra\R$  and an electric potential
$\gp:\cV\to \R$ are called $\G$-periodic  on $\cG$, if they satisfy
\[
\lb{per1} \a(\be+\ba)=\a(\be),\qq \gp_{x+\ba}=\gp_x \qq \forall \
(\be,\ba,x)\in \cA\ts\G\ts \cV.
\]
Let us recall the spectral properties  of the Schr\"odinger operator
$h_\a$ studied in \cite{KS17}, \cite{KS23}. We introduce the Hilbert space
\[\lb{Hisp}
\mH=L^2\Big(\T^{d},{dk\/(2\pi)^d}\,,\ell^2(\cV_*)\Big)
=\int_{\T^{d}}^{\os}\ell^2(\cV_*)\,{dk \/(2\pi)^d}\,, \qqq
\T^d=\R^d/(2\pi\Z)^d,
\]
i.e., a constant fiber direct integral, equipped with the norm
$
\|g\|^2_{\mH}=\int_{\T^d}\|g(k)\|_{\ell^2(\cV_*)}^2\frac{dk}{(2\pi)^d}\,,
$
where the function $g(k)\in\ell^2(\cV_*)$ for almost all $k\in\T^d$.
The parameter $k$ is called the \emph{quasimomentum}. We identify
the vertices of the fundamental graph $\cG_*=(\cV_*,\cE_*)$ with the
vertices of the $\G$-periodic graph $\cG=(\cV,\cE)$ in the
fundamental cell $\Omega$.
In \cite{KS17}, we have shown that there
exists a unitary operator $U:\ell^2(\cV)\to\mH$ such that
\[
\lb{raz}
\begin{aligned}
& Uh_\a U^{-1}=\int^\oplus_{\T^d}h_\a(k){dk\/(2\pi)^d},\qq
h_\a(k)=\D_\a(k)+\gp
\end{aligned}
\]
and the fiber Schr\"odinger operator $h_\a(k)$ and the fiber
magnetic Laplacian $\D_\a(k)$ act on $\ell^2(\cV_*)$. Here
$\D_\a(k)$  is some $\n\ts\n$ matrix with entries analytic in
$k\in \C^d$, where $\n:=\# \cV_*\in \N$, $\#M$ being the number
of elements in a set $M$. Each fiber operator $h_\a(k)$ has $\n$
real eigenvalues labeled in non-decreasing order (counting
multiplicities) by
\[
\label{eq.3H} \l_{1}(k)\leq\l_{2}(k)\leq\ldots\leq\l_{\nu}(k), \qqq
\forall\,k\in\T^{d}.
\]
 Each $\l_j(\cdot)$
is a real and piecewise analytic function  on the torus $\T^{d}$ and
gives rise to the \emph{spectral band} $\s_j(h_\a)$ given by
\[\lb{ban.1H}
\s_j(h_\a)=[\l_j^-,\l_j^+]=\l_j(\T^{d}),\qqq j\in\N_\n, \qqq
\N_\n=\{1,\ldots,\n\}.
\]
Then the spectrum of the operator  $h_\a$ has a band structure and
is a union of a finite number of bands given by
\[\lb{spec}
\s(h_\a)= \cup_{j=1}^{\nu}\s_j(h_\a),\qqq
\s_{j}(h_\a)=[\l_{j}^-,\l_{j}^+].
\]
Some of $\l_j(\cdot)$ may be constant, i.e.,
$\l_j(\cdot)=\L_j=\const$, on some subset $\cO$ of $\T^d$ of
positive Lebesgue measure. In this case the operator $h_\a$ on $\cG$
has the eigenvalue $\L_j$ of infinite multiplicity. We call
$\{\L_j\}$ a \emph{flat band}.  Thus, the spectrum of the magnetic
Schr\"odinger operator $h_\a$ on the periodic graph $\cG$ has the
form
\[
\lb{sfg} \s(h_\a)=\s_{ac}(h_\a)\cup \s_{fb}(h_\a),
\]
where $\s_{ac}(h_\a)$ is the absolutely continuous spectrum, which
is a union of non-degenerate bands, and $\s_{fb}(h_\a)$ is the set
of all flat bands $\{\L_j\}$. An open interval between two
neighboring non-degenerate bands is called a \emph{spectral gap}. It
is possible that $\s(h_\a)=\s_{fb}(h_\a)$ for specific
magnetic fields $\a$ and $\gp$, see recent results \cite{KS23}.
 Moreover, it was shown that for any periodic graph and any periodic
vector magnetic potential $\s_{ac}(\D_{g\a}+\gp)=\es$ only for a
finite number of coupling constants $g$ running through any finite
interval. Thus below we assume that $\s_{ac}(h_\a)\ne \es$.

 We discuss the existence and completeness of wave
operators for  $ h(t)=\D_{\b(t)}+V(t)$ in two situations. In the
first case, in Subsection \ref{sccateringintimeperiodc} we consider
Hamiltonians on $\Z^d$ with time-periodic electro-magnetic
potential, decaying at infinity of $\Z^d$, under the conditions
${\rm MZ_s}$ and ${\rm VZ_s}$ with $s = p, a$. In the second case, in
Subsection \ref{sec:1} we consider Hamiltonians on more general
graphs with time-periodic magnetic potential and time-independent
bounded electric potential, and in addition,  time-periodic electric
potential decaying at space infinity, under the conditions M and V. Our main results are Theorem \ref{Two2a}, Theorem \ref{T1} and Corollary \ref{CorGaugetransform}.

\subsection{Scattering in time-periodic systems}
\label{sccateringintimeperiodc}
 We consider two classes of
perturbations for magnetic and electric fields. Let $ \gB_p$ and
$\gB_a$ be the sect of functions $b(x)$ $: \Z^d \to {\bf R}$ such that
\[
\label{dp1}
\begin{aligned}
& b\in \gB_p \Longleftrightarrow  b^2(x)\in \ell^p(\Z^d),\qq \   d\ge 3,  \qq \
\ca 1\le p<{6\/5} & {\it if} \ \ d=3,\\
1\le p<{4\/3} &{\it  if}  \  \ d\ge 4, \ac
 \\
& b\in \gB_a \Longleftrightarrow   (1 + |x|)^a\, b^2(x) \in \ell^\iy(\Z^d), \qq a>1, \qq  d\ge 1.
\end{aligned}
\]
Throughout the paper, we denote a function $q$ on ${\R}\times
{\Z}^d$ as $q = q(t,x) = q_x(t)$. Our first class of magnetic,
electric potentials are assumed to satisfy the following conditions
for either $s =  p$ or $s = a$.

\medskip
\noindent {\bf Condition ${\bf MZ_s}$} : {\it The magnetic potential
$\b(\be,t)$ on $\cE$  is  $\t$-periodic in time, i.e.,
$\b(\be,t)=\b(\be,t+\t)$ for all $(\be,t)\in \cA\ts\T_\t$. Moreover
there exists $b \in \gB_s$ such that
\[
\lb{eg3m}
\begin{aligned}
|\b(\be,t)|\le 2b_xb_y\qqq \forall \ \be=(x,y).
\end{aligned}
\]
}

\noindent {\bf Condition ${\bf VZ_s}$} : {\it The electric potential
$V$ is written as $V =v+q$, where $v,q\in L^\iy(\T_\t\ts \Z^d)$ and
there exists $b\in \gB_s$ such that  letting $Q_x(t) := \int_0^t
q_x(s)ds$,  we have  for all $x\in \Z^d$ and ${\bf e} = (x,y)$
\[
\label{V1z}
\begin{aligned}
|v_x(t)|\le b_x^2,
\end{aligned}
\]
\[
\label{V2z}
\begin{aligned}
& Q_x(\t)=0, \qqq
\sup_{t\in\T_\t}|Q_x(t)|=o(1)\qq \as \qq |x|\to \iy,
\\
& |Q_{y}(t)-Q_x(t)|\le b_yb_x.\qq
\end{aligned}
\]
}

\bigskip
Our first main theorem is as follows. Let $U(t,s)$ be the evolution
operator for $h(t) = \Delta_{\beta (t)} + V(t)$, i.e. $U(t,s)f$ is
the unique solution to the initial value problem
\begin{equation}
i\frac{d}{dt}u(t) = h(t)u, \quad  t \in {\R}, \quad
u(s) = f.
\end{equation}

\begin{theorem} \label{Two2a} Assume that  the magnetic
potential $\b$ and the electric  potential $V = q + v$ of
$h(t)=\D_{\b(t)}+V(t)$ on $\ell^2(\Z^d)$
are $\t$-periodic and satisfy

\noindent (1) either "Conditions ${\rm MZ_p}$  and Conditions ${\rm
VZ_p}$",

\noindent (2) or "Conditions ${\rm MZ_a}$  and Conditions ${\rm
VZ_a}$".

\noindent
 Then the wave operator
\[
\begin{aligned}
\label{wo1}
   W_\pm=s-\lim U(0,t) e^{-it\D}\qqq  as\qqq t\to \pm\iy,
\end{aligned}
\]
 exists and is complete, i.e. $\Ran W_\pm=\cH_{ac}(U(\t,0))$.
\end{theorem}

\begin{remark}
 1) Theorem \ref{Two2a} describes the continuous spectrum of
 the monodromy  operators $U(\t,0)$. The discrete spectrum of $U(\t,0)$
in possible gaps in the continuous spectrum,  embedded eigenvalues
in the continuous spectrum and their numbers are discussed in \cite{K25}.

\no 2) Eigenvalues of $U(\t,0)$ for complex potentials are studied in \cite{K21}.
\end{remark}

As an example,  let us  consider a potential $V = v$ decaying sufficiently rapidly in $x$.

\begin{example}
 Consider the Hamiltonian $h(t)=\D+v(t)$ on $\Z^d, d\ge 1$, where
 $v(t)={A\/ 1+|x+\sin t|^a},  a>1$ and $A=\const$, $|x|^2=x_1^2+...+x_d^2$.
 As the  potetial $v$ satisfies Condition ${\rm VZ_a}$, by Theorem \ref {Two2a},
 the wave operator $W_\pm$ in \er{wo1} exists and is complete.
\end{example}

We can also deal with the potential $V=q$, where $q$ is not decaying but  oscillates sufficiently fast.

\begin{example} Consider the Hamiltonian $h(t)=\D+q_x(t)$ on $\Z^d, d\ge 1$,
where the potential $q_x(t)=A\cos (|x|^{2d}t)$ and $A=\const$,
 $|x|^2=x_1^2+...+x_d^2$. Note that $q$ is only oscillating in $t$ and
 $x$, but its integral
$$
Q_x(t)=\int_0^t q_x(s)ds=   {A\sin (|x|^{2d}t)\/|x|^{2d}}
$$
decays in $x$, since  $Q_x(t)={O(1)\/|x|^{2d}}$ as $|x|\to \iy$.
Thus $q$ satisfies Condition ${\rm VZ_a}$,  and by Theorem \ref{Two2a},
the wave operator $W_\pm$ in \er{wo1} exists and is complete.
\end{example}

\subsection{Scattering in space-time periodic systems}\label{sec:1}
We consider in a more general situation, and study a magnetic Schr\"odinger equation on a graph
$\cG$:
\[
\lb{eg1}
\begin{aligned}
{\frac{ d }{d t}}u(t)= -i h(t) u(t),\qq h(t)=\D_{\b(t)}+\gp+ V(t),
\end{aligned}
\]
where $h(t)$  is a discrete magnetic  Laplacian given by \er{1} and
$\gp:\cV\to \R $ is a bounded electric potential.  Let $\wt
\cH=L^2(\T_\t,\cH)$, $\T_\t=\R/(\t \Z)$. Also let $h_0 = \Delta$ be
the free Laplacian on $\cG$, and $R_0(\lambda) = (h_0 - i
\partial_t - \lambda)^{-1}$ on $\wt \cH$. Let $\wt\cB_\iy$ be the
set of compact operators on $\wt \cH$. We assume that the magnetic
potential  $\b$ is $\t$ -periodic in time and satisfies:

\medskip
\no{\bf Condition M.} {\it The electric potential $\gp$ is
time-independent and in $ \ell^\iy(\cV)$.  The magnetic potential
$\b(\cdot,t)$ on $\ell^\iy(\cE)$ has a form
$\b(\be,t)=\a(\be)+\d(\be,t)$, where $\a$ and $\delta$ are magnetic
potentials and $\d$ is $\t$-periodic in time (i.e.,
$\delta (\be,t) =\d(\be,t+\t)$ for all $(\be,t)\in \cA\ts\T_\t$) and
satisfies:
\[
\lb{eg3}
\begin{aligned}
\int_0^\t\sum_{\be\in \cE}  |\sin \d(\be,t)|dt<\iy.
\end{aligned}
\]
} We impose the following condition  on  the electric  potential
$V(t) = V(t,x)$ in (\ref{eg1}).

\medskip
\no {\bf Condition V.} {\it The function $V=v+q$, where  $v,q$ are
$\t$-periodic and satisfy:
\[
\label{V1}
\begin{aligned}
v\in L^1(\T_\t,\ell^1(\cV)), \qqq q\in L^1(\T_\t,\ell^\iy(\cV)),
\end{aligned}
\]
and the potential $Q_x(t)=\int_0^t q_x(s)ds$ for $(x,t)\in \cV\ts
\R$ satisfies
\[
\label{V3}
\begin{aligned}
Q(\t)=0,\qq  Q\in L^2(\T_\t,\ell^2(\cV)),
\end{aligned}
\]
\[
\label{V4}
\begin{aligned}
  \sum_{\be=(x,y)\in\cA}|Q_{x+\be}(t)-Q_x(t)|\in
L^1(\T_\t,\ell^1(\cV)).
\end{aligned}
\]
}

The second main theorem is as follows.
\begin{theorem}
\lb{T1}
 Let Schr\"odinger operators $h(t)=\D_{\b(t)}+\gp+V(t)$  and
$h_\a=\D_{\a}+\gp$ act on the Hilbert  space $\cH=\ell^2(\cV)$,
where the magnetic potential $\b=\a+ \delta$ satisfies Condition M, the
electric  potential $\gp\in \ell^\iy(\cV)$ is real and the potential
$V$ satisfies Condition V. We assume that $\s_{ac}(h_\a)\neq \es$.
Then the wave operator
\[
\lb{Wo} W_\pm=s-\lim U(0,t)e^{-ith_\a}P_{ac}(h_\a)\qqq \as \ t\to
\pm \iy
\]
 exists and is complete, i.e. $\Ran W_\pm=\cH_{ac}(U(\t,0))$.
 \end{theorem}

We apply Theorem  \ref{T1} to periodic graphs with periodic magnetic
vector potentials $\a:\cA\ra\R$  and periodic electric potentials
$\gp:\cV\to \R$. Recall the spectral structure \er{sfg}.

\begin{corollary}
\label{CorGaugetransform}
\lb{T1c} Let $\cG=(\cV, \cE)$ be a periodic graph.
 Let $h(t)=\D_{\b(t)}+\gp+V(t)$ act on the Hilbert  space
$\cH=\ell^2(\cV)$, where the magnetic potential $\b=\a+\delta$
satisfies Condition M, the potential $\gp\in \ell^\iy(\cV)$ and the
potential $V$ satisfies Condition V. We assume that the Schr\"odinger
operator $h_\a=\D_{\a}+\gp$, where the magnetic potential $\a$ and
the electric potential $\gp$ are  real $\G$-periodic, satisfies
$\s_{ac}(h_\a)\neq \es$. Then the wave operator
\[
\lb{WoC} W_\pm=s-\lim U(0,t)e^{-ith_\a}P_{ac}(h_\a)\qqq \as \ t\to
\pm \iy
\]
 exists and is complete, i.e. $\Ran W_\pm=\cH_{ac}(U(\t,0))$.
 \end{corollary}

As is well-known, for time-independent potentials $V(x)$, there is a
sharp border line of the decay rate for the existence and
completeness of wave operators. The situation is more complicated
for time-dependent perturbations. For the sake of comparison, in
Theorem 6.1, we discuss the case of electric and magnetic potentials
decaying $t$ and $x$. It is also well-known that in some cases
strongly oscillating perturbations (in $t$ or $x$) form an
exceptional class. We shall also discuss them.

Although the above theorems are apparently similar, the method of
proof is different. Theorem \ref{Two2a} is proven by Kato-Kuroda's
stationary scattering theory. For the proof of Theorems \ref{T1}, we
use time-dependent scattering theory, especially Howland's theory of
time-dependent Hamiltonian and the gauge transform.

We can also apply Theorem \ref{T1} to different graphs, for example,
to trees in which the corresponding Laplacian has an absolutely continuous spectrum. Such
Laplacians on trees are studied by Aptekarev,
 Denisov, Yattselev in their paper \cite{ADY20}.

Different properties of
Schr\"odinger operators with periodic potentials  on periodic graphs
were considered in \cite{KS14, MW89}, see also
references therein. Schr\"odinger operators with periodic magnetic
and electrostatic potentials  are discussed in \cite{KS17, KS23}.

\subsection{Related works}
We briefly describe results about Schr\"odinger operators with
decaying potentials on periodic graphs. The scattering problem for
the Schr\"odinger operator with a decaying potential on the lattice
$\Z^d , d > 1$, was considered in the papers \cite{BS99, IK12, IM14,
KM19}, see also references therein. The discrete spectrum was
discussed in \cite{BSL18,DHKS,KS20,RoS09}. The discrete Schr\"odinger
operator with decaying potentials on arbitrary periodic graphs was
studied in \cite{KS20}, \cite{PR18}. The Schr\"odinger operator with
a potential periodic in some directions and finitely supported in
others on arbitrary periodic graphs was investigated in the article
\cite{KS17}. Trace formulas and eigenvalues estimates  for discrete Schr\"odinger operators
were discussed in \cite{IK12, KL18, K17}.

 Scattering theory for discrete Schr\"odinger operators on graphs with
long-range potentials were discussed in \cite{N14, Ta19,Ta19m,
Ta20}. Inverse scattering theory for the discrete Schr\"odinger
operators with finitely supported potentials was considered in
\cite{IK12} for the case of the lattice $Z^d$ and in \cite{A13} for
the case of the hexagonal lattice.

While there is a lot of results on
time periodic Hamiltonians on $\R^d, d\ge 1$,  \cite{AK19,AKS10,H74,H79,K84,K85,K89,M00,S75,Y77}
 including the time periodic magnetic fields \cite{AK16,K19,K89},
 there are only few results about time periodic Hamiltonians on the graphs \cite{K21,K25,OP18}.

\subsection{Acknowledgements}
H. I. is indebted to the support by Grant-in-Aid for Scientific Research (C) 24K06768.

\section{Propagators} \setcounter{equation}{0}
\subsection{Quasienergy operators}

Let $\mathcal H$ be a separable Hilbert space.
 Throughout the paper
by a {\it propagator} we mean a two-parameter family of unitary
operators $U(t,s)$, $t, s \in {\R}$, acting on $\mathcal H$ and
satisfying the following conditions:

\medskip
{\it
$\bu$ $U(t,s)U(s,r)=U(t,r)$ for all $t,s,r\in \R$,

$\bu$ $U(t,t)=I$ for all $t\in \R$,

$\bu$ $U(t,s)$ is strongly continuous in $t,s\in \R$.}


\medskip

Now we discuss the scattering for time periodic Hamiltonian
$h(t), t\in \R$. We assume that
$h(t)$ is $\t$-periodic in time, i.e., $h(t + \t) = h(t)$ for any $t\in \R$.
Letting $\T_\t=\R/(\t \Z)$, we introduce the
space $\wt \cH=L^2(\T_\t,\cH)$ of  $\cH$-valued functions ${\mathbb R} \ni t \to f(t)$ that are
$\t$--periodic in $t$ equipped with norm
$$
\|f\|^2_{\ti
\cH}={1\/\t}\int_0^\t \|f(t)\|^2_{\cH}dt.
$$
For the propagator
$U(t,s)$ on $\cH$, the mapping
$$
\widetilde U(s) :\wt\cH\ni f\to (\widetilde U(s)f)(t) := \,
U(t,t+s)f(t+s),\qq (t,s)\in \T_\t\ts \R
$$
defines a strongly continuous group on $\wt\cH$  (with respect to the parameter $s\in \R$). By
Stone's theorem,  this group defines a self-adjoint  operator $\wt
h$ acting on $\wt\cH$ as
\[
\begin{aligned}
  \label{U2}
   (e^{is \wt h}f)(t)=U(t,t+s)f(t+s), \qq f\in \wt\cH.
\end{aligned}
\]
We call such an operator $\wt h$ a {\it  quasienergy  operator} and the
spectrum of $\wt h$ the {\it  quasienergies}. We consider the case
when $U(t,s)$ is the propagator for a family of $\t$-periodic
self-adjoint operators $h(t)$. In the case when $h(t)=h_0+V(t), t\in
\R$, $U(t,s)$ is constructed by the standard method of iteration for
the integral equation (see \er{p2} and \cite{H79}). We introduce a
notation. For a $t$-dependent, $\tau$-periodic  bounded operator
$A(t)$ on $\mathcal H$, we define the operator $\langle A(t)\rangle
$ on $\wt\cH $ by
\begin{equation}
\langle A(t)\rangle : \wt\cH \ni f(t) \to A(t)f(t) \in \wt\cH.
\end{equation}
For the propagator $U(t,s)$,  the monodromy operator is defined by
\begin{equation}
M(t)=U(t+\t,t), \qqq t\in \R.
\label{M(t)monodromyoperator}
\end{equation}
The above conditions for the propagator yield
\begin{equation}
M(t)=U(t,0)U(\t,0)U(0,t).
\end{equation}
Evaluating at $s=-\t$ in (\ref{U2}), as $f \in \widetilde{\mathcal
H}$ is $\tau$-periodic, we have
\[
  \label{U3}
  e^{-i\t\wt h}=\langle M(t)\rangle.
\]
 Let $\pa=-i\frac{\pa}{\pa t}$ be the
self-adjoint operator in $L^2(\T_\t)$. The same notation
$\pa=-i\frac{\pa}{\pa t}$ is used for the corresponding operator in
$\wt \cH$ with the natural domain $\mD(\pa)$.

\begin{example}
For a bounded self-adjoint operator  $h_0$ on $\cH$, define $\wt
h_0=\pa +h_0$ on $\wt\cH$. Then $\wt h_0$ is the quasienergy
operator associated with the propagator $U_0(t,s)=e^{-i(t-s) h_0}$,
 which is $\tau$-periodic for any $\tau > 0$,  and satisfies
\[
\begin{aligned}
  \label{U4}
   (e^{is \wt h_0}f)(t)=U_0(t,t+s)f(t+s)=e^{is h_0}f(t+s), \qq f\in \wt\cH,
\qq t,s\in \R.
\end{aligned}
\]
Furthermore, we have:
\begin{itemize}
\item
 If  $h_0=0$, then  $U_0(t,s)=I$, $\wt
h_0=-i{\pa\/\pa t}$, and $e^{is \wt h_0}$ is the group of
translation $(e^{is \wt h_0}f)(t)=f(t+s)$.
\item
If we know the spectrum $\s(h_0)$, then the spectrum of $\wt h_0$
has the form
\[
\s(\wt h_0)=\bigcup_{n\in \Z} \s(h_0+\o n),\qqq \o={2\pi\/\t}.
\]
Note that if $\o$ is  big enough and if $\s(h_0)$ is a finite
interval, then the spectrum of $\wt h_0$ has the band structure with
the bands $\s(\D+\o n)$ separated by gaps.
\end{itemize}
\end{example}

\medskip
Let $\cB_1$ and $\cB_2$ be the trace class and the Hilbert-Schmidt
class of operators in $\cH$, equipped with  norm $\|\cdot
\|_{\cB_1}$ and $ \|\cdot \|_{\cB_2}$, respectively. Let $\cB_\iy$
be the class of compact operators in $\cH$. The corresponding
classes of operators acting in $\wt\cH$ are denoted by $\wt\cB_1,
\wt\cB_2, \widetilde{\cB}_\iy.$

Given a bounded self-adjoint operator  $h_0$ on $\cH$ and $\wt
h_0=\pa +h_0$ on $\wt\cH$, we introduce the free resolvent
$R_0(\l)=(\wt h_0-\l)^{-1}, \l\in \C_\pm$. We will show that if
$V=\lan V(t) \ran$ satisfies Condition V in (\ref{V1})-(\ref{V4}) or
Condition H in (\ref{v3})-(\ref{v5x}), then
\[
\lb{1.c} VR_0(\l)\in \wt\cB_\iy, \qq \forall \ \l\in \C_\pm.
\]
This yields  $\wt h=\wt h_0+V$ is self-adjoint on $\mD(\wt
h)=\mD(\wt h_0)$. Define the perturbed resolvent $R(\l)=(\wt
h-\l)^{-1}, \l \in \C_\pm$.
We have the useful identity
\[
\lb{i1} \lan e^{it\o}\ran R(\l)\lan e^{-it\o}\ran  =R(\l+\o),\qq
\forall \l\in \C_\pm.
\]
It means that the spectrum of $\wt h$ (and $\wt h_0$)  is
$\o$--periodic.

\bigskip
We recall the well-known results for the  quasienergy
operators, see e.g., \cite{H74, H79}.
Let $U(t,s)$ be the propagator for a $\t$-periodic family of
operators $h(t)$, and $\wt h$ the corresponding quasienergy
operator. Then

{\it
\begin{itemize}
\item
 The following identities hold true:
\[
\lb{2.5} P_{r}(\wt h)=\langle P_{r}(M(t))\rangle,\qq {\rm for}  \ \
r ={ac}, \ {pp}, \ {sc}, \ {c}, \ {fb},
\]
where $M(t)$ is the monodromy operator defined in
(\ref{M(t)monodromyoperator}) and $P_{r}(H)$ is the projection
associated with the spectral subspace for a self-adjoint operator
$H$ with ac = absolutely continuous, pp = pure point, sc = singular
continuous, c = continuous and fb = flat band.

\item
 If  $\p$ is the eigenvector of the quasienergy operator $\wt
h$, i.e.,  $\wt h\p=\l \p$ for some $\l\in\R$, then $\p$ is strongly
continuous in $t\in \R$ and satisfies
\[
\begin{aligned}
\label{u1x}
  \p(t)=e^{it \l }U(t,0)\p(0), \qqq t\in \R,
\end{aligned}
\]
 and consequently,  $\p(t)$ is the eigenvector of the monodromy operator:
$$
M(t)\p(t)=e^{-i\t\l }\p(t),\qq t\in \T.
$$
Moreover, the vector  $\p_n=\langle e^{i\o n t}\rangle\p$, where
$\o={2\pi\/\t}$,  is the eigenvector of the  operator $\wt h$ and
\[
\label{u1xcc}
 \wt h\p_n=(\l+\o n )\p_n,
\]
where $
n\in  \Z$ and $0\le \l<\o$.
 Any eigenvector  of $\wt\cH_{pp}(\wt h)$ has this  form $\p_n$.  Similarly any vector  in
$\wt\cH_{pp}(M(t)), t\in \T$, has  the  form $\p(t)$ of
\eqref{u1xcc}, where $0\le \l<\o$ and  $\wt h\p=\l\p$.

\item
 Let $n\in \Z$. Then
\[
\lb{2.7} \wt h \langle e^{i\o n t}\rangle =\langle e^{i\o n
t}\rangle (\wt h+\o n).
\]
\item
 Let $U_0(t,s)$ be the propagator for a $\t$-periodic family of
operators $h_0(t)$, and let $\wt h_0$ be the corresponding
quasienergy operator. Define the wave operators by
\[
\lb{Wo12}
\begin{aligned}
W_\pm=s-\lim U(0,s)e^{-ish_0}P_{ac}(h_0),
\\
\wt W_\pm=s-\lim e^{is \wt h} e^{-is \wt h_0}P_{ac}(\wt h_0)
\end{aligned}
\]
as  $s\to \pm \iy$. Then $W_\pm$ exists and is
complete, i.e. $W_\pm \cH=\cH_{ac}(U(\t,0))$ if and only if
$\wt W_\pm$ exists and is complete, i.e. $\wt W_\pm\wt \cH=\cH_{ac}(\wt
h)$.
\end{itemize}
}

\subsection{Operators on graphs}
We consider some properties of operators on graphs  $\cG=(\cV,\cE)$.
Define the space $\ell^p(\cV), p\ge 1$ of all  sequences
$f=(f_x)_{x\in \Z^d}$ equipped with  norm
$$
\begin{aligned}
\|f\|_{p}=\|f\|_{\ell^p(\cV)} =\big(\sum_{x\in
\cV}|f_x|^p\big)^{1\/p},\qq  \ p\in [1,\iy),
\end{aligned}
$$
and $\|f\|_{\iy}=\|f\|_{\ell^\iy(\cV)} =\sup_{x\in \cV}|f_x|$. For a
Banach space $\gB$, the space $\ell^p(\gB)$ is defined similarly
with $\cV$ replaced by $\gB$. Let $L^r(\T_\t,\gB)$ be the space of
$\gB$-valued $L^r$ functions on $\T_\t$. In the case $f(\cdot)\in
L^r(\T_\t,\ell^p(\cV))$ we define  the norm $\|f\|_{r,p}$ by
\[
\|f\|_{p,r}^r=\int_{\T_\t}\|f(t)\|_{\ell^p(\cV)}^rdt, \qq p,r \ge 1.
\]
Note that $\|f\|_{p,r}\le \|f\|_{s,r}$ for all $p\ge s\ge 1$.  For
$\cH=\ell^2(\cV)$, the space $\ti \cH=L^2(\T_\t,\cH)$ can be
realized as $\ell^2(\cH)$ via the Fourier transform $\F: \ti \cH\to
\ell^2(\cH)$ defined by
$$
\begin{aligned}
f\to \F f=(f_n)_{n\in \Z}, \qq  f_n=(\F f)_n={1\/\sqrt \t}\int_0^\t
e^{-int\o}f(t)dt,\qq \o={2\pi\/\t},\qq  f\in \ti \cH.
\end{aligned}
$$
  We use the notation  $\lan A(t)\ran$ to denote the operator of
multiplication  by $A(t)$ on $\wt\cH$. Introduce the operators $\wt
h_0$ and $\wt h$ on $\wt\cH=L^2(\T_\t, \ell^2(\cV))$ by
$$
\wt h_0=\pa+\D,\qqq  \wt h=\wt h_0 +\lan V(t) \ran,
$$
where $\Delta$ is the Laplacian (\ref{DefineDeltaforV}) on $\mathcal
V$.  The spectrum of $\wt h_0$ has the form
\[
\s(\wt h_0)=\s_{ac}(\wt h_0)=\bigcup_{n\in \Z} \s(\D+\o
n)=\bigcup_{n\in \Z} [\s(\D)+\o n].
\]
 For a function $f_x, x\in \cV,$  we have an
estimate
\[
\lb{ii}
\begin{aligned}
 \sum_{x\in\cV}
\sum_{\be=(x,y)\in\cA}|f_y|^2\le \vk_+\|f\|_{\ell^2(\cV)}^2,
\end{aligned}
\]
Consider time periodic magnetic perturbations with magnetic
potential $\b=\a+\d(\be,t)$.

\begin{lemma} \label{Tmp}
Let $F(t)=\D_{\b(t)}-\D_{\a},  t\in\T_\t$ with
magnetic potential $\b(\be,t)=\a(\be)+\d(\be,t)$.

\no i) If $
 |\d(\be,t)|\le b_xb_y$ for all  $\be=(x,y)\in \cA$,
 for some $b\in \ell^p(\cV), p\ge 1$. Then
\[
\lb{mp1x}  F(t)=b F_b(t)b,\qqq  \|F_b(t)\|_{\cB}\le \vk_+.
\]
\no ii) If $\d(\be,t)$  satisfies Condition M. Then
\[
\lb{mp1V}   F\in L^1(\T_\t,\cB_1).
\]

\end{lemma}
\begin{proof}
 Let  $s(\be,t)=\sin {\d(\be,t)\/2}$ for the edge
$\be=(x,y)\in\cA$.   From $\b=\a+\d$ we obtain
\[
\lb{ML} (F(t) f)_x=\tfrac{1}{2} \!\!\!\!
\sum_{\be=(x,y)\in\cA}\!\!\!\!\big(e^{i\a(\be)}-e^{i\b(\be,t)}\big)f_y
=-i\!\!\!\! \sum_{\be=(x,y)\in\cA}\!\!\!\!
e^{i\a(\be)+i{\d(\be,t)\/2}} s(\be,t)f_y.
\]
i) This and the estimate $|\d(\be,t)|\le b_xb_y$ yield
$$
\begin{aligned}
|(F_b(t)f)_x| \le\sum_{\be=(x,y)\in\cA}|b_x^{-1}s(\be,t)b_y^{-1}|
|f_y|\le
 \sum_{\be=(x,y)\in\cA}|f_y|
\le \sqrt{\vk_+} \rt(\sum_{\be=(x,y)\in\cA} |f_y|^2\rt)^{1\/2},
\end{aligned}
$$
by the condition (\ref{kappaxdefine}) ${\vk}_x\le \vk_+$. Then the
estimate \er{ii} implies
$$
\begin{aligned}
 \|F_b(t)f\|^2\le  
 \vk_+\sum_{x\in\cV} \sum_{\be=(x,y)\in\cA} |f_y|^2
\le \vk_+^2 \|f\|^2,
\end{aligned}
$$
 which  yields \er{mp1x}.

 \no ii)   From \er{ML} we obtain
$$
 \|F(t)\|_{\cB_1}\le  \sum_{x\in\cV}  \sum_{\be=(x,y)\in\cA}\!\!\!\!
\|s(\be,t)\|_{\cB_1}.
$$
 which  yields \er{mp1V}.
\end{proof}

We need a simple fact about the operator $\pa=-i{\pa\/\pa t}$ on
$L^2(\T_\t, \ell^2(\cV))$.

\begin{lemma} \label{Tmpx}

i) Let $h_0$ be a bounded self-adjoint operator on the Hilbert space
$\ell^2(\cV)$ and let a function $q\in L^2(\T_\t,\ell^p(\cV))$ for
some $p\in [1,\iy)$. Then
\[
\lb{mp1xxx}   q(\pa+h_0 +i)^{-1} \in \wt\cB_\iy.
\]
ii)  Let a function $Q\in L^1(\T_\t,\ell^p(\cV))\cap
L^\iy(\T_\t,\ell^\iy(\cV))$ for some $p\in [1,\iy)$. Then
\[
\lb{mp1xx}   (e^{iQ}-I)(\pa +h_0+i)^{-1} \in \wt\cB_\iy.
\]

\end{lemma}
\begin{proof} The proof of  i) is given in \cite{K25}.

ii)  Due to the Tailor  series for $e^{-iQ}$   this function has the
properties $e^{-iQ}-I=-iQ+Q_\bu$, where $Q_\bu\in
L^1(\T_\t,\ell^p(\cV))$. This together with i) yields \er{mp1xx}.
\end{proof}

\section{Scattering on Lattices}
\setcounter{equation}{0}

\subsection {Preliminaries}
In this section, we discuss a scattering theory for  time-periodic
Schr\"odinger equations on the lattice $\Z^d$ and prove Theorem
\ref{Two2a}. We consider a $\t$-periodic (in time) system
\[
\lb{01}
\begin{aligned}
{\frac{ d }{d t}}u(t)= -i h(t) u(t),\qq h(t)=\D_{\b(t)}+ V(t),
\end{aligned}
\]
where,
$\D_{\b(t)}$ is the magnetic Laplacian given by
\[
\big(\D_{\b(t)} f\big)_x=\frac{1}{2}\sum_{|x-y|=1}(f_x-
e^{i\b(\be,t)}f_y), \qqq f=(f_x)_{x\in{\Z}^d} \in
\ell^{2}({\Z}^d),
\]
and the potential $V(t)$ is $\t$-periodic in time: $(V(t)
f)_x=V_x(t)f_x,$ for all $(t,x)\in \R\ts \Z^d$.  Introduce the
operators $\wt h_0$ and $\wt h$ on $\wt\cH=L^2(\T_\t, \ell^2(\Z^d))$
by
$$
\wt h_0=\pa+\D,\qqq  \wt h=\wt h_0 +\lan V(t) \ran,
$$
where $\Delta$ is the Laplacian (\ref{DefineDeltaforV}) for
$\mathcal V = {\mathbb Z}^d$. Recall that the spectrum
$\s(\D)=\s_{\textup{ac}}(\D)=[0,2d]$. Then the spectrum of $\wt h_0$
has the form
\[
\s(\wt h_0)=\s_{ac}(\wt h_0)=\bigcup_{n\in \Z} \s(\D+\o
n)=\bigcup_{n\in \Z} [\o n,\o n+2d].
\]
Note that if $\o >2d$, then the spectrum of $\wt h_0$ has the band
structure with the bands $\s(\D+\o n)=[\o n,\o n+2d]$ separated by
gaps.   Introduce the free resolvent $ R_0(\l)=(\wt h_0-\l)^{-1},
\l\in \C_\pm$. Below we show that if $V$ satisfies Condition $VZ_s$
for $s = p$ or $s = a$, then
\[
\lb{1.cx} VR_0(\l)\in \wt\cB_\iy, \qq \forall \ \l\in \C_\pm.
\]
This yields $\mD(\wt h)=\mD(\wt h_0)$. Define the perturbed
resolvent $R(\l)=(\wt h-\l)^{-1} $ for all $\l\in \C_\pm\sm
\s_{disc}(\wt h)$. In the following, we sometimes refer to descriptions in later sections, since they are standard results and hold in
general situations.

\begin{lemma}
\lb{Ta2} i) Let $h_0$ be a bounded self-adjoint operator on the
Hilbert space $\cH$ and $g\in \ell^p(\Z^d), p\in [1,\iy)$. Then
\[
\lb{com}
g(\pa +h_0-\l)^{-1}\in \wt\cB_\iy\qqq \l\in \C_\pm.
\]
ii) Let $\b$ and $V$ satisfy Conditions $MZ_s$ and $VZ_s$ in
Subsection \ref{sccateringintimeperiodc} respectively. Then for
$h(t)=\D_\b(t)+V(t)$ there exists a propagator $U(t,s)$ such that
$\pa +\lan h(t)\ran$ is a quasi-energy operator and
\[
\lb{ooq}   (e^{i\s \wt h}f)(t) =U(t,t+\s) f(t+\s), \qqq \forall \
t,\s\in \R.
\]
\end{lemma}
\begin{proof}
We have  \er{com} from \er{mp1xxx}. To show ii) we rewrite
$h(t)=\D_\b(t)+V(t)$ in the form $h(t)=\D_\a+ Y(t)$, where
$Y(t)=\D_\b(t)-\D_\a+V(t)$. By the condition of the lemma, we obtain
$Y\in L^1(\T_\t,\cB)$ and thus due to Lemma \ref{Ta1} the statement
ii) holds true.
\end{proof}

 In order to treat the potential $q$ we define a unitary
operator $J(t)=e^{-iQ(t)}, t\in \R$, where $Q_x(t)=\int_0^t
q_x(s)ds$ and $Q\in L^\iy(\T_\t\ts \Z^d)$. By using a gauge
transformation $H=J^* \wt h J$ we have:
\[
\begin{aligned}
H=J^* (\pa +\D+(F+v+q))J=\wt h_0+T,\qq T=X+Y,
\end{aligned}
\]
where
\[
\lb{Q1}
\begin{aligned}
F=\D_{\b}-\D,\qq X=J^* (F+v) J,\qq Y=e^{iQ} \D e^{-iQ}-\D,
\end{aligned}
\]
and
\[
\begin{aligned}
\lb{Q2}
(e^{iQ} h_0 e^{-iQ}f)_x &=
{e^{iQ_x}\/2} \sum_{\be=(x,y)\in\cA}\big(e^{-iQ_x}f_x-e^{-iQ_y}f_y\big)=(\D_\f f)_x,
\\
&(\D_\f f)_x={1\/2}\sum_{\be=(x,y)\in\cA}\big(f_x-e^{i\f(\be,t)}f_y\big).
\end{aligned}
\]
Here the edge $\be=(x,y)$ and  $\D_\f$ is the new magnetic Laplacian with
  the new magnetic field $\f(\be,t)=Q_x(t)-iQ_y(t)$
 after the gauge transformation $J$.

\begin{lemma} \label{TLz}
Define the operators $F(t)=\D_{\b(t)}-\D$ and $Y(t)=e^{iQ(t)} \D
e^{-iQ(t)}-\D, t\in\T_\t$, where the functions $\b, q$ satisfy
Condition $MZ_s$ and $VZ_s$ respectively. Then $F, Y$ satisfy
\[
\lb{mp1}  F(t)=bF_b(t)b, \qq  \|F_b(t)\|_{\cB}\le d,
\]
\[
\lb{mp2}  Y(t)=bY_b(t)b, \qqq
\|Y_b(t)\|_{\cB}\le d.
\]
\end{lemma}
\begin{proof} From \er{mp1x} we obtain \er{mp1}.
 In order to show \er{mp2} we use \er{Q2} with  a
new magnetic field $\f(\be,t)=Q_x(t)-iQ_y(t), \be=(x,y)\in\cA$ after
the gauge transformation $J$.
    Condition $MZ_s$  implies that the new magnetic field satisfies
$$
|\f(\be,t)|=|Q_y(t)-Q_x(t)|\le b_x b_y, \qqq \forall \ \be=(x,y)\in\cA.
$$
Then from \er{Q1}, \er{Q2} and  \er{mp1} we obtain \er{mp2}.
\end{proof}

\subsection {Theory of Kato-Kuroda}
We need the following results of Kato and Kuroda \cite{KK71}. We
consider self-adjoint operators $H_0$ and $H=H_0+V$ on a separable
Hilbert space $\cK$, where the perturbation $V$ satisfies

\medskip
\noindent
{\bf Condition KK.}  {\it Let $H_0$ and $H=H_0+V$ be self-adjoint
operators on a separable Hilbert space $\cK$ and $V(H_0-i)^{-1}\in
\cB_\iy$. Let $V=\x V_\x\x$, where the operator $\x\ge 0$ and the
operator $V_\x$ is bounded. Let $\Bbbk \ss \s_{ac}(H_0)$ be a
finite interval and there is a factorization
 \[
 \begin{aligned}
 \gF(\l)=V_\x \x(H_o-\l)^{-1}\x,\qqq  \l\in \C_\pm,
 \end{aligned}
 \]
with the following properties:

\begin{itemize}
\item  the operator $V(H_0-i)^{-1}\in \cB_\iy$,

\item the function $\gF$ is analytic in $\K_\pm=\Bbbk\ts \R_\pm\ss
\C_\pm$ and uniformly H\"older in $\ol \K_\pm $,
\item each $\gF(\l), \l\in G_\pm$ is compact.
\end{itemize}
}

\medskip
We formulate  a local version of the Kato-Kuroda Theorem (see \cite{KK71}
Theorem 7.1,  p 164)  in the form convenient for us.

\begin{theorem}
\lb{TA97} ({\bf Kato-Kuroda \cite{KK71}}).
 Let $H_0$ and $H=H_0+V$ satisfy Condition KK for a compact interval
 $\Bbbk\ss \s_{ac}(H_0)$. Then

i) There is a closed set $\cE\ss \Bbbk$ with Lebesgue measure zero
such that the operator $I+\gF(\l\pm i0), \l\in \Bbbk\sm \cE$ is
invertible.

ii)   The set $\Bigl(\sigma_\mathrm{pp}(H)\cup
\sigma_\mathrm{sc}(H)\Bigr)\cap \Bbbk\ss \cE$.

iii) Each eigenvalue $\lambda\in\sigma_\mathrm{pp}(H)\cap \Bbbk$ is
of finite multiplicity.

iv) The wave operator
\[
W_\pm (\Bbbk)=s-\lim e^{itH}e^{-itH_0}E_0(\Bbbk)P_{ac}(H_0)\qqq \as
\ t\to \pm \iy
\]
 exists and is complete, i.e. $\Ran W_\pm (\Bbbk)=\cK_{ac}(H,\Bbbk)$
 and $\s_{sc}(H)\cap \Bbbk\ss \cE$.

 \end{theorem}

We apply the Kato-Kuroda Theorem to time periodic Hamiltonians $
h(t)=\D+V(t), t\in \R$ on $\Z^d$, where $\D$ is the Laplacian. We
assume that the potential $V$ satisfies Condition $ VZ_s$ for $s =
p$ or $a$. In order to apply the Kato-Kuroda Theorem we use the
approach of Howland \cite{H79} to study the operators $\wt
h_0=\pa+\D$ and $\wt h=\wt h_0+V$ on $\wt\cH=L^2(\T_\t,
\ell^2(\Z^d))$.
 We need results about the resolvent
$r_0(\l)=(\D-\l)^{-1}$ on $\ell^{2}(\Z^d)$ from \cite{IK12}.

\begin{theorem}
\lb{Tsp2} Let $\r_a=(1+|x|)^{-a}$ for  $x\in \Z^d,d\ge 1$ and $a >{1\/2}$. Define
$$
\gS = {\bf C} \setminus [0,2d],
\quad
\gS_\d = {\gS}\sm \bigcup_{j=0}^d \{|\l-2j|\le \d\},
$$
and  an operator-valued function $\gf: \gS_\d\to \cB$ by
$$
\gf(\l)=\r_a r_0(\l)\r_a,\qq \l \in \gS_\d
$$
for some $\d>0$.  Then $\gf$  is analytic in $\gS_\d$, H\"older up to the boundary, and satisfies
\[
\label{r1}
\begin{aligned}
 \|\gf(\l)\|_{\cB}\le {C_{d,a, \d}\/\max \{1,\gr(\l)\}},\qqq \forall \ \ \l\in\gS_\d,
\end{aligned}
\]
where  $C_{d,a,\d}$ is a constant depending on $d,a,\d$ only and
$\gr(\l)=\dist(\l,\s(\D))$.

If, in addition, $a>1,d\ge 3$, then $\gf$  is analytic in $\gS_\d$,
H\"older up to the boundary, and satisfies
\[
\label{r13}
\begin{aligned}
 \|\gf(\l)\|_{\cB}\le {C_{d,a}\/\max \{1,\gr(\l)\}},\qqq \forall \ \ \l\in\gS_\d,
\end{aligned}
\]
where  $C_{d,a}$ is a constant depending only on $d,a$.
\end{theorem}

We need weighted estimates for the resolvent
$r_0(\l)=(\D-\l)^{-1}$ on $\ell^{2}(\Z^d)$ from \cite{KM19}.

\begin{theorem}
\lb{Tsp3} Let $b\in \ell^p(\Z^d), d\ge 3$ where $p$ satisfies
\er{dp1}. Define an operator-valued function $\gf: \gS_\d \to \cB$
by $\gf(\l)=b r_0(\l)b,\ \l \in \gS_\d$. Then $\gf$ is analytic in
$\gS_\d$, H\"older up to the boundary, and satisfies
\[
\label{r2}
\begin{aligned}
 \|\gf(\l)\|_{\cB_2}\le {C_{d,p}\|b\|_p^p\/\max \{1,\gr(\l)\}} ,\qqq \forall \ \
 \l\in \gS_\d,
\end{aligned}
\]
where  $C_{d,p}$ is some constant depending only on $d,p$ and
$\gr(\l)=\dist(\l,\s(\D))$.
\end{theorem}

Due to the periodicity of the resolvent \er{i1}:
$e_j^*R(\l)e_j=R(\l-\o j)$  for all $(j,\l)\in \Z\ts \C_\pm$, it is
enough to consider any bounded interval $\Bbbk\ss [0,\o]\cap \s(\wt
h_0)$. Define the sets of thresholds  $\s_t(\wt h_0)$ for $\wt h_0$ on the interval
$[0,\o]$ and the set  $\wt \gS_\d$ by
$$
 \s_t(\wt h_0)=[0,\o]\cap \bigcup_0^{2d} (2j+\o\Z),\qqq
\wt \gS_\d^\pm=\{\Re \l \in (0,\o): \dist (\l, \s_t(\wt h_0)) >\d \}\cap \C_\pm.
$$
Letting $R_0(\l)=(\pa+\D-\l)^{-1}$ and using Theorem \ref{Tsp2} we
obtain

\begin{corollary}
\lb{Tz2} Let $ a>{1\/2}, \d>0$ and $d\ge 1$. Define  an
operator-valued function $\wt\gf: \wt \gS_\d^\pm\to \wt\cB$ by
$\wt\gf(\l)=\r_a R_0(\l)\r_a$, for $ \l \in\wt \gS_\d^\pm$. Then
$\wt\gf$  is analytic in $\wt \gS_\d^\pm$, H\"older up to the
boundary, and satisfies
\[
\label{r12x}
\begin{aligned}
 \|\wt\gf(\l)\|_{\wt\cB}\le C_{d,a, \d},\qqq \forall \ \ \l\in\gS_\d,
\end{aligned}
\]
where  $C_{d,a,\d}$ is a constant depending on $d,a,\d$ only.

\end{corollary}

By Theorem  \ref{Tsp3}, we obtain

\begin{corollary}
\lb{Tz3} Let $b\in \ell^p(\Z^d), d\ge 3$ where $p$ satisfies
\er{dp1}. Define  an operator-valued function $\wt\gf: \C_\pm\to
\cB$ by $\gf(\l)=b R_0(\l)b,\ \l \in \C_\pm$. Then $\wt\gf$  is
analytic in $\C_\pm$, H\"older up to the boundary, and satisfies
\[
\label{r12xw}
\begin{aligned}
 \|\wt\gf(\l)\|_{\wt\cB}\le C_{d,p},\qqq \forall \ \ \l\in\C_\pm,
\end{aligned}
\]
where  $C_{d,p}$ is a constant depending only on $d,p$.
\end{corollary}

Using Theorem  \ref{TA97}  and Corollary
\ref{Tz2}, \ref{Tz3}   we give

\medskip
\noindent \no {\bf Proof of Theorem \ref{Two2a} for the case $V =
v$.} Let $h(t)=\D_{\b(t)}+V(t)$ on $\cH=\ell^2(\Z^d), d\ge 1$. We
consider the case in which the magnetic potential $\b$ and the
potential $V$ satisfy Condition $MZ_a$ and $VZ_a$ respectively.
The proof for the case $s = p$ is
similar. We rewrite our Hamiltonian $h(t)$ in the form
$$
h(t)=\D_{\b(t)}+V(t)=\D+X(t), \qq X=F+v, \qq F(t)=
\D_{\b(t)}-\D.
$$
We consider an unperturbed quasienergy operator $\wt h_0=\pa +h_0$
on $\wt\cH=L^2(\T_\t,\cH)$, and rewrite the perturbed quasienergy
operator in the form $\wt h=\wt h_0+X$ on $\wt\cH$. The operators
$\D_{\b(t)}, v(t)$ are bounded. Therefore, the operator $\wt h$ is
self-adjoint on the domain $\mD$.

Form Condition $MZ_a$ and Lemma \ref{TLz} we deduce that $F=bF_bb$,
 where $\|F_b(t)\|_{\cB}\le d$. This and Condition $VZ_a$ gives
\[
\lb{Xbb}
X=bX_{b}b, \qqq X_b=F_{b}+b^{-2}v\in \cB(\wt\cH).
\]
We check the Condition KK for the operators  $\wt h_0=\pa+\D$
and $\wt h=\wt h_0+X$ on $\wt\cH$. We define an operator-valued function
$$
\gF(\l)=X_{b}b R_0(\l)b,\qqq  \l\in \C_\pm.
$$

\begin{itemize}
\item  The relation \er{com} gives that the operator  $XR_0(\l)\in
\cB_\iy(\wt\cH), \Im \l\ne 0$.
\item
 Let $\Bbbk \ss \s_{ac}(\wt h_o)\cap \ol\gS_\d^+$ be a
compact interval for some small  $\d>0$.  By Corollary \ref{Tz2},
the function $\gF(\l)$ is analytic in $\gS_\d^\pm$ and is continuous
up to the boundary. Then the first condition  in Condition KK holds
true.

\item the relation \er{com} gives that the operator
 $\gF(i)\in \cB_\iy$.
\end{itemize}

Then Condition KK for the operators $H,\wt h_0$ then holds true and
by the Kato-Kuroda Theorem \ref{TA97}, the wave operators $W_\pm
(\wt h,\wt h_0)$ exists and is complete. In addition, we have

\medskip
i)  There is a closed set $\cE\ss \s(\wt h_0)$  with Lebesgue
measure zero such that the operators $I+\gF(\l\pm i0), \l\in \s(\wt h_0)\sm \cE$ are
invertible.

ii)   The set $\sigma_\mathrm{pp}(\wt h)\cup \sigma_\mathrm{sc}(\wt h)\ss
\cE$.

iii) Each eigenvalue of $\wt h$ is of finite multiplicity.
\BBox

\medskip

\no {\bf Proof of Theorem \ref{Two2a} for the case $V = q + v$.} Let
$h(t)=\D_{\b(t)}+V(t)$ on $\cH=\ell^2(\Z^d), d\ge 1$, where the
magnetic potential $\b$ and the potential $V$ satisfy Condition
$MZ_s$ and $VZ_s$ respectively. We consider an unperturbed
quasienergy operator $\wt h_0=\pa +h_0$ and a perturbed quasienergy
operator $\wt h=\pa +\lan h(t)\ran$ on $\wt\cH$. The operators
$\D_{\b(t)}, V(t)$ are bounded.  Hence,  the operator $\wt h$ is
self-adjoint on the domain $\mD$.  In order to treat the potential
$q$ we define a gauge transformation $J(t)=e^{-iQ(t)}, t\in \R$,
where $Q_x(t)=\int_0^t q_x(s)ds$ and $Q\in L^\iy(\T_\t\ts \Z^d)$. We
make a unitary transformation $H=J^* \wt h J$ and have the
following:
\[
\begin{aligned}
H=J^* (\pa +\lan \D_{\b(t)}+V(t)\ran)J=\pa + J^* \lan
\D_{\b(t)}\ran)J+v =\pa +\lan \D_{\vp(t)}\ran)+v,
\\
(J^* \D_{\b(t)} J
f)_x={1\/2}\sum_{\be=(x,y)\in\cA}\big(f_x-e^{i\vp(\be,t)}f_y\big)=(\D_{\vp(t)}f)_x,
\end{aligned}
\]
where $\vp(\be,t)=\b(\be,t)+Q_x(t)-Q_y(t)$ and the edge $
\be=(x,y)$. The new magnetic potential $\vp$ satisfies Condition
$MZ_s$ and the new potential $v$ satisfies Condition $VZ_s$. Then by
Theorem \ref{Two2a}, for the case $q = 0$, the wave operators $W_\pm
(H,\wt h_0)$ exist and are complete.

Consider the wave operators $\wt W_\pm=W_\pm(\wt h,\wt h_0)$ for
$\wt h,\wt h_0$ defined by
\[
\lb{w1} \wt  W_\pm= s-\lim e^{i\s \wt h}e^{-i\s\wt h_o}\qqq \as \qq \s\to \pm \iy.
\]
We have $|J_x(t)-\1|\le \n_x$, for some bounded function $\n_x, x\in
\Z^d$  such that $\n(x)=o(1)$ as $|x|\to \iy$. Then from \er{com} we
obtain $\n R_0(i)\in \wt\cB_\iy$  and   \er{J1} yields
\[
\lb{w2} \wt W_\pm= s-\lim e^{i\s \wt h}Je^{-i\s\wt h_0}
=s-\lim J e^{i\s H}e^{-is\wt h_0}=JW_\pm(H,\wt
h_0)\ \ \as \qq \s\to \pm \iy.
\]
 Thus the properties of the wave operators $W_\pm(\wt h,\wt h_0)$ are equivalent
  to those of  $W_\pm(H,\wt h_0)$. Then  due to Theorem \ref{Two2a},
  for the case $q = 0$, the  wave operators $W_\pm(\wt h,\wt h_0)$
 exist and are complete, i.e. $\Ran W_\pm(\wt h,\wt h_0)=\wt\cH_{ac}(\wt h)$. In
 addition, we have

i) There is a closed set $\cE\ss \s(\wt h_0)$ with Lebesgue
measure zero such that the operators $I+\gF(\l\pm i0), \l\in \s(\wt h_0)\sm \cE$ are
invertible.

ii)   The set $\sigma_\mathrm{pp}(\wt h)\cup \sigma_\mathrm{sc}(\wt h)\ss
\cE$.

iii) Each eigenvalue of $\wt h$ is of finite multiplicity.
\BBox

\section{Time-dependent approach}
We study the existence and completeness of wave operators  from a
different framework.
\textcolor{blue}{In this section, we devote ourselves to the abstract case, where the operator $h(t)$ can be   integral operators for instance.
The main theorem  is Theorem
\ref{T2}.  We apply it for scattering on graphs in the next section.}

We introduce a condition on perturbations with zero average for the period.

\medskip
\no {\bf Condition H.} {\it 1) Let $\mathbb R \ni t \to V(t)=V^*(t)$ has
the form $V=v+q$, where $v,q$ satisfy:
\[
\label{v3} q\in L^1(\T_\t,\cB),\qq  v\in L^1(\T_\t,\cB_1).
\]
 2)  Let $q$ and $Q(t):=\int_0^t q(s)ds$ and a bounded self-adjoint operator
$h_0$  satisfy
\[
\label{v5}
\begin{aligned}
  Q(\t)=0,\qqq  Q\in L^2(\T_\t,\cB_2),
\end{aligned}
\]
\[
\label{v5x}
\begin{aligned}
\int_0^t q(s)Q(s)ds\in L^\iy(\T_\t,\cB_1),\qqq  (Qh_0-h_0Q)\in
L^1(\T_\t,\cB_1).
\end{aligned}
\]
}
\medskip
We are going to prove
\begin{theorem}
\lb{T2}
 Let $h(t)=h_0+V(t), t\in \T_\t$ act on the separable Hilbert  space
$\cH$,  where the operator $h_0$ is bounded self-adjoint and the
operator-valued function $V(\cdot): \R\to \cB(\cH)$ is $\t$-periodic
and satisfies Condition H. Then the wave operator
\[
W_\pm=s-\lim U(0,t)e^{-ith_0}P_{ac}(h_0)\qq \as \qq  t\to \pm \iy
\]
 exists and is complete, i.e. $W_\pm\cH=\cH_{ac}(U(\t,0))$.
 \end{theorem}

We discuss the existence of  a propagator $U(t,s), t,s\in \R$ for a
Hamiltonian $h(t)=h_0 +V(t)$ on a Hilbert space $\cH$, where $h_0$
is a bounded self-adjoint operator on $\cH$ and $V(t)$ belongs to
$L^1(\T_\t,\cB(\cH))$ for some $\t>0$. It is well known that the
propagator $U(t,s)$ is determined in the usual way by iteration of
the integral equation
\[
\lb{p2} U(t,s)=U_0(t,s)-i\int_s^t U_0(t,r)V(r)U(r,s)dr,
\]
or another integral equation
\[
\lb{p2x} U(t,s)=U_0(t,s)-i\int_s^t U(t,r)V(r)U_0(r,s)dr,
\]
where $ U_0(t,s)=e^{-i(t-s)h_0}$. The solution of this equation has
the form
\[
\lb{p3}
\begin{aligned}
& U(t,s)=U_0(t,s)+\sum_{j=1}^\iy (-i)^j U_j(t,s), \qq
\\
&  U_j(t,s)= \int_s^t \int_s^{t_1}...
\int_s^{t_{j-1}}U_0(t,t_1)V(t_1)U_0(t_1,t_2)...U_0(t_{j-1},t_j)V(t_j)U_0(t_{j},s)dt_1...dt_j.
\end{aligned}
\]
We recall the well-known results from \cite{H79}.

\begin{lemma}
\lb{Ta1} Let $h_0$ be a bounded self-adjoint operator on a Hilbert
space $\cH$. Let an operator-valued function $t\to V(t)=V^*(t)$
belong to $L^1(\T_\t,\cB(\cH))$ for some $\t>0$. Then there exists a
unique propagator $U(t,s), t,s\in \R$, which satisfies  the integral
equation \er{p2} and has the form \er{p3} where the series converges
uniformly in every bounded subsets in $(t,s,V)\in \R^2\ts
L^1(\T_\t,\cB(\cH))$ and has an estimate
\[
\lb{p3x} \|U(t,s)-U_0(t,s)\|\le \min\{A(t,s),1\}e^{A(t,s)}, \qqq
t,s\in \R,
\]
where $A(t,s)=\int_s^t \|V(\s)\|d\s$.
Moreover, $\wt h=\wt h_0+V$ is the quasienergy operator and
satisfies
\[
\begin{aligned}
\label{qe1}
   (e^{is \wt h}f)(t)=U(t,t+s)f(t+s), \qq f\in \wt\cH, \qq t,s\in \R,
\end{aligned}
\]
\[
\begin{aligned}
\label{qe1x} (e^{iz \wt h}e^{-iz \wt h_0}-I)f(t) =i \int_t^{t+z}
U(t,r)V(r)U_0(r,t)f(t) dr,\qqq t,z\in \R.
\end{aligned}
\]
\end{lemma}

\no {\bf Proof.} We will show \er{p3x} only, since all the other
results  were obtained in \cite{H79}. Consider the $j$-th term
$U_j(t,s)$:
\begin{eqnarray*}
\|U_j(t,s)\|  & \le& \int_s^t \int_s^{t_1}...
\int_s^{t_{j-1}}\|U_0(t,t_1)V(t_1)U_0(t_1,t_2)...U_0(t_{j-1},t_j)V(t_j)U_0(t_{j},s)\|dt_1...dt_j
\\
& \le& \int_s^t
\int_s^{t_1}...\int_s^{t_{j-1}}\|V(t_1)\|...\|V(t_j)\|dt_1...dt_j
 \\
& =& {1\/j!}\int_s^t
\int_s^{t}...\int_s^{t}\|V(t_1)\|...\|V(t_j)\|dt_1...dt_j={1\/j!}
\rt(\int_s^t \|V(t_1)\|dt_1\rt)^j.
\end{eqnarray*}
This shows that for each $t,s\in \R$ the series \er{p3} converges
uniformly and absolutely on every bounded subsets in $
L^1(\T_\t,\cB(\cH))$. Summing them up, we obtain \er{p3x}. \BBox

\medskip
In order to prove Theorem \ref{T2} we recall standard facts, see
e.g., \cite{K21}.

\begin{proposition}
\lb{Tre2}
\no i)
 Let $h_0$ be a  bounded self-adjoint operator  on a separable
  Hilbert space $\cH$. Then,
the resolvent $R_0(\l)=(\pa+h_0-\l)^{-1}, \l\in \C\sm \R $ has the form
\[
\lb{ddx1}
\begin{aligned}
R_0(\l)f(t)=i e^{it \vp}\int_0^\t \rt(\1_{t-s}+ {e^{i\t
\vp}\/1-e^{i\t \vp}}\rt) e^{-is\vp} f(s)ds, \quad  f\in L^2(\T_\t,\cH),
\end{aligned}
\]
where  $\vp=\l-h_0$ and  $\1_t=1, t>0$, $\1_t=0, t<0$.

\no ii) In addition,  for any operator-valued function $V\in
L^2(\T_\t,\cB_2(\cH))$, we have
\[
\lb{res2} V(\pa -i)^{-1}\in \wt\cB_2.
\]
\end{proposition}

We need the following fact for gauge transformations.

\begin{lemma} \label{TJ}
Let $h$ be self-adjoint on a Hilbert space $\cH$. Then for a bounded
operator $a$ on $\cH$ such that $a(h-i)^{-1}$ is compact, we have
for $f\in \cH_{ac}(h)$:
\[
\lb{J1} \|ae^{is h}f\|\to 0 \qqq \as \qq s\to \pm \iy.
\]
 Let, in addition,   $A\in
L^1(\T_\t,\cB_2)\cap L^\iy(\T_\t,\B)$ and $h$ be  bounded. Then
\[
\lb{J2} \|A(s)e^{is h}f\|\to 0 \qqq \as \qq s\to \pm \iy.
\]
\end{lemma}
\no \begin{proof} The fact \er{J1} is well-known, see e.g.,
\cite{RS75}. We show \er{J2}. Consider operators $A=\lan A(t)\ran$
and $\wt h=\pa +h$ acting in $\wt H$. Let $\|\cdot \|=\|\cdot
\|_{\wt \cH}$ and let $\lan \cdot,\cdot \ran$ be the scalar product
in $\cH$. Due to \er{res2} the operator $A(\pa +h-i)^{-1}\in \wt
\cB_\iy$. Then \er{J1} give for $ s\to \pm \iy$:
$$
\|e^{i\s \pa}A e^{i\s \wt h}f\|_{\wt \cH}^2=\int_0^\t \|A(t+\s)
e^{i\s h}f\|^2dt=\int_0^\t F(t,s)dt
 \to 0,
$$
where $s=t+\s$ and $F(t,s):=\|A(s)e^{is h}e^{-it h}f\|^2$. This
yields $F(t,s)\to \gF(t)$ as $s\to \pm \iy$ for almost all $t\in
[0,\t]$, where $\gF(t)=0$ for almost all $t\in [0,\t]$. We show that
$\gF(t)$ is continuous. We prove it at the point $t=0$, the proof for the
other points is similar. Let $(s,t)\in \R\ts [0,\t]$. Define the operators
$$
X(s,t)=g(s)e^{it h}- e^{it h}g(s),\qq  g(s)=e^{-is h}A^*(s)A(s)e^{is
h}.
$$
Using the identity
$$
\begin{aligned}
F(0,s)-F(t,s)=\lan X(s,t)e^{-it h}f,f \ran=\lan g(s) (e^{-it
h}-I)f,f\ran   -\lan  (e^{it h}-I)g(s)e^{-it h}f,f \ran
\\
\end{aligned}
$$
we obtain
$$
\begin{aligned}
|F(0,s)-F(t,s)|\le 2\|f\|\|(e^{-it h}-I)f\| \sup_{t\in
(0,\t)}\|A(t)\|^2.
\\
\end{aligned}
$$
If $F(t,s)\to 0$ as $s\to \pm \iy$, then $|F(0,s)-F(t,s)|$ is small
for small $t$, which yields that  $\gF(0,s)=o(1)$ as $s\to \pm \iy$.
\end{proof}

We use one more  property  of a  propagator $U(t,s), t,s\in \R$ for  time
periodic  Hamiltonians.

\begin{lemma} \label{Tq1}
 Let an operator-valued function $q(\cdot): \R\to \cB(\cH)$ be such
 that
\[
\lb{gG1} q(t)=q^*(t), \quad \forall\ t\in \R,\qq q\in L^1(\T_\t,\cB),\qq
Q(\t)=0,
\quad J_2\in L^\iy(\T_\t,\cB_1),
\]
where $Q(t):=\int_0^t q(s)ds$ and $J_2(t):=\int_0^t q(s)Q(s)ds$.
 Then for the Hamiltonian  $q(t), t\in \R$ there exists  a propagator
$U(t,s), t,s\in \R$, which satisfies for all $t\in \T_\t$
\[
\lb{ug1} U(t,0)=I-iQ(t)+Q_\bu(t),
\]
\begin{equation}
  \label{u1}
 \|Q_\bu(t)\|_{\cB_1}\le Ce^{\int_0^t \|q(s)\|ds},\qqq \where \qq C=  \sup_{t\in
 \T_\t}\|J_2(t)\|_{\cB_1}.
\end{equation}
\end{lemma}
\begin{proof}
 From  Lemma  \ref{Ta1} we deduce that for $q(t),
t\in \R$ there exists  a propagator $U(t,s,q), t,s\in \R$ and
$J(t)=U(t,0,q), t\in \R$ having the form \er{ug1}, where
\[
\lb{1x1} Q_\bu(t)= \sum_{j=2}^\iy  J_j(t), \ \ J_{j+1}(t)=-i\int_0^t
q(r)J_j(r)dr,\ \ J_{3}(t)=-i\int_0^t q(s_1)J_2(s_1)ds_1,...
\]
Let $f_t=\|J_{2}(t)\|_{\cB_1}$ and $C=\sup_{t\in \T_\t}f_t$ and
$y_t=\|q(t)\|$. Then using \er{1x1} we have
$$
\begin{aligned}
& \|J_{3}(t)\|_{\cB_1}\le \int_0^t y_s f_sds\le C\int_0^t y_sds,
\\
&  \|J_{4}(t)\|_{\cB_1}\le C \int_0^t y_{s_1}ds_1 \int_0^{s_1}
y_{s_2}ds_2=
  {C\/2}\rt(\int_0^t y_{s_1}ds_1\rt)^2,
\\
&  \|J_{5}(t)\|_{\cB_1}\le  C \int_0^t y_{s_1}ds_1 \int_0^{s_1}
y_{s_2}ds_2 \int_0^{s_2} y_{s_3}ds_3
  ={C\/3!}\rt(\int_0^t y_{s_1}ds_1\rt)^3,....
\end{aligned}
$$
This shows that for each $t\in \R$ the series \er{1x1} converges
uniformly and absolutely on every bounded subsets in $(t,s,J_2)\in
\R^2\ts L^\iy(\T_\t,\cB_1(\cH))$.  Summing the majorants we obtain the
estimate \er{u1}.
\end{proof}

In order to discuss  scattering with time periodic trace class
perturbations $V(t)$ we recall the  following result of Howland
(p. 483 in \cite{H79}), see also Schmidt's paper \cite{S75}.

\begin{theorem}
\lb{THS} Let $h_0$ be a bounded self-adjoint operator acting on a
separable Hilbert  space $\cH$.  Let $V(t)$ be a measurable
self-adjoint function $V(\cdot): \R\to \cB_1(\cH)$ . Assume that
$V(t)$ has period $\t$ and $\int_0^\t \|V(t)\|_{\cB_1}dt<\iy$  and
let $U(t,s)$ be the propagator for  $h(t)=h_0+V(t), t\in \T_\t$.
Then the wave operators
\[
W_\pm=s-\lim U(0,t)e^{-ith_0}P_{ac}(h_0)\qqq \as \ t\to \pm \iy
\]
 exist and complete, i.e. $W_\pm\cH=\cH_{ac}(U(\t,0))$.
 \end{theorem}

Here the propagator $U(t,s)$  is obtained in the usual way by iteration
of the integral equation \er{p2}, where  $U_0(t,s)=e^{-i(t-s)h_0}$.
There is no problem for bounded operators with locally integrable norm.
The function $F(t,s)=U_0(s,t)U(t,s)$ satisfies
\[
F(t,s)=\1-i\int_s^t U_0(s,r)V(r)U(r,s)dr
\]
and is therefore locally continuous in norm.

\bigskip

 \no {\bf Proof of Theorem  \ref{T2}.} Let $h(t)=h_0+V(t), t\in
\T_\t$ on the separable Hilbert  space $\cH$,  where the operator
$h_0$ is bounded self-adjoint and the $\t$-periodic operator-valued
function $V=V^*=v+q$ satisfies Condition H. We have the propagator
$\cU(t,s,q)$ for $q(t)$ and let $J(t)=\cU(t,0,q)$. Due to Lemma
\ref{Tq1} the propagator $J(t)$ has the form
\[
\begin{aligned}
  \label{bb1z}
J(t)=I-iQ(t)+Q_\bu(t)\qqq  \forall\ \ t\in \R, \qqq Q_\bu\in
L^\iy(\T_\t,\cB_1).
\end{aligned}
\]
Due to conditions on $h(t)$ and Lemma \ref{Ta1} there exists the
propagator $U(t,s)$ for $h(t), t\in \R$. We make the unitary
transformation $\ul U(t,s)=J^*(t) U(t,s) J(s)$. Using \er{Uq}, we
obtain
\[
\begin{aligned}
  \label{uU}
& i{d\/dt}\Big(\ul U(t,s)f\big)=J^*(t)
\Big(-q(t)+h(t)\Big)U(t,s)J(t)f=\ul h(t)\ul U(t,s)f,
\\
& \ul h(t)=J^*(t)\Big(h_0+v\Big)J(t)=h_0+\ul V(t),\qq \ul
V(t)=V_1(t)+J^*v(t) J,
\\
& V_1=J^*h_0J-h_0 =J^*(iK+K_\bu), \qq
 \qq K=Q\D-\D Q,\qq  K_\bu=h_0Q_\bu-Q_\bu h_0.
\end{aligned}
\]
Thus  $\ul U(t,s)$ be the propagator for $\ul h(t), t\in \R$. From
Condition H and \er{bb1z} we deduce that $K, K_\bu\in L^1(\T_\t,
\cB_1)$. Then $\ul V\in L^1(\T_\t, \cB_1)$ and due to
Howland-Schmidt's Theorem \ref{THS}, the  wave operators
\[
\lb{Wop1aa} W_\pm(\ul h,h_0)=s-\lim \ul U(0,t)e^{\mp i
h_ot}P_{ac}(h_0)\qqq \as \ t\to \pm \iy
\]
 exist and are complete, i.e. $ W_\pm(\ul h,h_0)\cH=\cH_{ac}(\ul U(\t,0))$.
 In addition, \er{bb1z} , \er{J2}  imply
$$
\begin{aligned}
W_\pm(h,h_0)=s-\lim U(0,t)e^{\mp i h_0t}P_{ac}(h_0)=s-\lim
U(0,t)J(t) e^{\mp i h_0t}P_{ac}(h_0)
\\
=s-\lim \ul U(0,t) e^{\mp i h_0t}P_{ac}(h_0) =W_\pm(\ul h,h_0) \qqq
\as \ t\to \pm \iy.
\end{aligned}
$$
 Thus the wave operators $W_{\pm}(h,h_0)$ exist and complete, i.e.
${\rm Ran}\, W_{\pm}(h,h_0)=  \cH_{ac}(U(\t,0))$. Note that
 using $J(\t)=I$ we have the identity $\ul U(\t,0)=U(\t,0)$. \BBox

\begin{example}
 Let $h_0$ be a  bounded self-adjoint operator  on a separable
  Hilbert space $\cH$. Consider a hamiltonian $h(t)=h_0+q(t)$.
Let a $2\pi$-periodic perturbation $q(t)$ satisfy
\[
q(t)=\sum_{n\in \Z, n\ne 0}q_ne^{i2\pi nt}, \qq q_n\in \cB_1,\qq
\sum_{n\in \Z,n\ne 0}|n|^{-2\a}\|q_n\|_{\cB_1}^2  <\iy,\qq
\a<{1\/2}.
\]
 Then we have that $q\in L^2([0,2\pi],\cB_2)$  and
\[
Q(t)=\int_0^tq(s)ds=\int_0^t \sum_{n\in \Z, n\ne 0}q_ne^{i2\pi
ns}ds=\sum_{n\in \Z, n\ne 0}q_n {e^{i2\pi nt}-1\/i2\pi n}.
\]
Thus we obtain
$$\|Q(t)\|_{\cB_1}^2\le \sum_{n\in \Z,n\ne 0}|n|^{-2\a}\|q_n\|_{\cB_1}^2 \sum_{j\in \Z,
n\ne 0}{1\/\pi^2 j^{2-2\a}}<\iy,\qqq t\in \R.
$$
\end{example}

\section{Scattering on graphs}
\label{SectionScatteringongraph}
\setcounter{equation}{0}

\textcolor{blue}{From Theorem \ref{T2} we obtain the following result.
\begin{corollary}
\lb{Tc1}
 Let Schr\"odinger operators $h(t)=\D_{\b(t)}+\gp+V(t)$  and
 $h_\a=\D_{\a}+\gp$ act on the Hilbert  space
$\cH=\ell^2(\cV)$ and such that $\s_{ac}(h_\a)\neq \es$, where the
magnetic potential $\b=\a+\b_*$ satisfies Condition M, the electric
potential $\gp\in \ell^\iy(\cV)$ is real and the potential $V=v+q$,
where  $v,q$ are $\t$-periodic and satisfy:
$$
\begin{aligned}
v\in L^1(\T_\t,\ell^1(\cV)),\qq q\in L^1(\T_\t,\ell^\iy(\cV)),\qq
Q(\t)=0,\qq  Q\in L^2(\T_\t,\ell^2(\cV)),
\\
\int_0^t q(s)Q(s)ds\in L^\iy(\T_\t,\ell^1(\cV)),\qqq  (Q\D-\D Q)\in
L^1(\T_\t,\ell^1(\cV)),
\end{aligned}
$$
where $Q_x(t)=\int_0^t q_x(s)ds$ for $(x,t)\in \cV\ts \R$.  Then the
wave operator
$$
 W_\pm=s-\lim U(0,t)e^{-ith_\a}P_{ac}(h_\a)\quad  \as \quad  t\to \pm \iy
$$
 exists and is complete, i.e. $\Ran W_\pm=\cH_{ac}(U(\t,0))$.
 \end{corollary}
 \no {\bf Proof.} Conditions of the corollary imply Condition H used in Theorem \ref{T2}. \BBox
 }

 \medskip
 \textcolor{blue}{ In order to
 apply  Corollary \ref{Tc1} to scattering on graphs, we add  ome more improvement. The key idea is the
gauge transformation. In general cases, the gauge transformation is performed in terms of some propagator, which is not easy to control. In our case, however,  $V$ is a scalar potential, and we can compute the gauge transformation exactly as an integral. }

\medskip
\noindent
\no {\bf Proof of Theorem \ref{T1}.} We rewrite the Hamiltonian
$h(t)=\D_\b+\gp+V$ in the form
$$
h=\D_\b+\gp+V=h_0+q+v_\b,\qq h_0=\D_\a+\gp, \qq v_\b=\D_\b-\D_\a+v.
$$
From \er{V1} and \er{mp1V} we obtain
\[
v_\b\in L^1(\T_\t\ts \cB_1).
\]
In order to treat the potential $q$ we define a gauge transformation
 $J(t)=e^{-iQ(t)}, t\in \R$, where $Q_x(t)=\int_0^t q_x(s)ds$. Due to the Tailor
 series for $e^{-iQ}$   this function has the properties
\[
\begin{aligned}
  \label{Uq}
e^{-iQ(t)}=I-iQ(t)+Q_\bu(t)\qqq  \forall\ \ t\in \R, \qqq Q_\bu\in
L^1(\T_\t,\cB_1).
\end{aligned}
\]
Due to conditions on $h(t)$ and Lemma \ref{Ta1} there exists the
propagator $U(t,s)$ for $h(t), t\in \R$. We make the unitary
transformation $\ul U(t,s)=J^*(t) U(t,s) J(s)$. Using \er{Uq}, we
obtain
\[
\begin{aligned}
  \label{uU1}
& i{d\/dt}\Big(\ul U(t,s)f\big)=J^*(t)
\Big(-q(t)+h(t)\Big)U(t,s)J(t)f=\ul h(t)\ul U(t,s)f,
\\
& \ul h(t)=J^*(t)\Big(h_0+v_\b\Big)J(t)=h_0+\ul V(t),\qq \ul
V(t)=V_1(t)+J^*v_\b(t) J,
\\
& V_1=J^*h_0J-h_0 =J^*(iK+K_\bu), \qq
 \qq K=Q\D-\D Q,\qq  K_\bu=h_0Q_\bu-Q_\bu h_0.
\end{aligned}
\]
Thus  $\ul U(t,s)$ be the propagator for $\ul h(t), t\in \R$. From
\er{V4} and \er{Uq}, \er{uU1} we deduce that $K, K_\bu\in L^1(\T_\t,
\cB_1)$ and then $\ul V\in L^1(\T_\t, \cB_1)$.
 Then due to Howland-Schmidt's Theorem \ref{THS}, the  wave operators
\[
\lb{Wop1aa1} W_\pm(\ul h,h_0)=s-\lim \ul U(0,t)e^{\mp i
h_0t}P_{ac}(h_0)\qqq \as \ t\to \pm \iy
\]
 exist and are complete, i.e. $ W_\pm(\ul h,h_0)\cH=\cH_{ac}(\ul U(\t,0))$.
 In addition, \er{Uq} , \er{J2}  imply
$$
\begin{aligned}
W_\pm(h,h_0)=s-\lim U(0,t)e^{\mp i h_0t}P_{ac}(h_0)=s-\lim
U(0,t)J(t) e^{\mp i h_0t}P_{ac}(h_0)
\\
=s-\lim \ul U(0,t) e^{\mp i h_0t}P_{ac}(h_0) =W_\pm(\ul h,h_0) \qqq
\as \ t\to \pm \iy.
\end{aligned}
$$
 Thus the wave operators $W_{\pm}(h,h_0)$ exist and complete, i.e.
${\rm Ran}\, W_{\pm}(h,h_0)=  \cH_{ac}(U(\t,0))$. Note that
 using $J(\t)=I$ we have the identity $\ul U(\t,0)=U(\t,0)$. \BBox

\section{ Scattering  for potentials decaying in time}
\setcounter{equation}{0}

We consider a scattering problem for a magnetic Schr\"odinger equation
$$
\begin{aligned}
{\tfrac{ d }{d t}}u(t)= -i h(t) u(t),\qq h(t)=\D_{\b(t)}+ V(t),
\end{aligned}
$$
on a lattice $\Z^d, d\ge 1$, where the Hamiltonian $h(t) $ is
time-dependent in the following sense: $\D_{\b(t)}$ is a discrete
magnetic  Laplacian with time depending magnetic potential  $\b$ on
$\Z^d$. The  time-dependent electric potential $V(t)$ and magnetic
potential  $\b$ satisfy.

\medskip
\noindent {\bf Condition R} :{\it The function  $V$ has the form $V
=v+q$, where $v,q\in L^\iy(\R\ts \Z^d)$ and $Q_x(t):=\int_0^t
q_x(s)ds$.  The electric potential $V(t)$ and the
magnetic potential $\b(\be,t)$ on $\cE$ depends on time and
satisfies for all edge $\be=(x,y), |x-y|=1$
\[
\lb{rm}
\begin{aligned}
& \sup_{|t|\ge 1}|Q_x(t)|=o(1)\qq \as \qq |x|\to \iy,
\\
& |v_x(t)|\le  w(t)F_x,
\\
& |\b(\be,t)|+|Q_{y}(t)-Q_x(t)|\le  w(t)F_y,
\end{aligned}
\]
where $F_y=c\vr_a(x)+ b_y$,
for some $(w,b)\in L^2(\R)\ts\ell^p(\Z^d)$ and $a>1$ and a constant $c$
and $p$ satisfies \er{dp1}.}

\medskip

Our main theorem for systems decaying in time  is as follows.

\begin{theorem} \label{TRwo} Consider the Hamiltonian $h(t)=\D_{\b(t)}+V(t)$
on $\ell^2(\Z^d)$, $d\ge 3$ and assume that  the magnetic potential
$\b$ and the electric  potential $V = q + v$ satisfy Condition R.
Then the wave operator
\[
\begin{aligned}
\label{r}
   W_\pm=s-\lim U(0,t) e^{-it\D}\qqq  as\qqq t\to \pm\iy,
\end{aligned}
\]
 exists and is  unitary.
\end{theorem}

\begin{remark}
Compared to Theorems  \ref{Two2a} and \ref{T1}, Theorem \ref{TRwo}
deals with the case when the potential $v$ is decaying in both $x\in
\Z^d$ and $t\in \R$.
\end{remark}

As an example,  let us  consider a potential $V = v$ decaying sufficiently rapidly in $x$.

\begin{example}
 Consider the Hamiltonian $h(t)=\D+v(t)$ on $\Z^d, d\ge 3$, where
 $$
 v(t)={A w(t)\/ 1+|x+\sin t|^a},  a>1, \qqq w\in L^2(\R),\qqq  A=\const,
 \qq |x|^2=x_1^2+...+x_d^2.
  $$
Since the potetial $v$ satisfies Condition R, by Theorem \ref
{TRwo},  the wave operator $W_\pm$ in \er{r} exists and is unitary.
\end{example}

We can also deal with the potential $V=q$, where $q$ is not decaying but
  oscillates sufficiently fast.

\begin{example} Consider the Hamiltonian $h(t)=\D+q_x(t)$ on $\Z^d, d\ge 3$,
where the potential $q_x(t)=A\cos (|x|^{a}t^\g), a>1$ and $A=\const,
\g>{3\/2}$. Note that $q$ is only oscillating in $t$ and $x$, but
the function
$$
Q_x(t)=\int_0^t q_x(s)ds=   {A\sin (|x|^{a}t^\g)\/\g t^{\g-1}|x|^{a}}+...,
$$
decays in $x, t$, since  $Q_x(t)={O(1)\/|x|^{2d}}$ as $|x|\to \iy$.
Thus $q$ satisfies Condition R and by Theorem \ref{TRwo},
the wave operator $W_\pm$ in \er{r} exists and is unitary.
\end{example}


\medskip
\no {\bf Proof of Theorem \ref{TRwo}.} Consider the operator
$h(t)=\D_{\b(t)}+V(t)$ on $\Z^d$, where the magnetic potential $\b$ and
the electric potential $V$ satisfy Condition R.

Firstly we discuss the case $q=0$.
From \er{p2x} we have
\[
\lb{R1} U(0,t)e^{-it\D}f=f-i\int_0^t U(0,s)X(s)e^{-is\D}fds,
\qq X(s)=\D_{\b(s)}-\D+v.
\]
We need results from \cite{KK71}. If there is an estimates \er{r2}, \er{r13}, then estimates
\[
\lb{R2}  \int_\R (\|\vr_a e^{-is\D}f\|^2+\|b e^{-is\D}f\|^2)ds\le
C_o^2\|f\|^2,\qq \forall \ f\in \ell^2(\Z^d),
\]
hold true, where $C_o$ does not depend on $f$.

 If  the edge
$\be=(x,y)\in\cA$, where $x,y\in \Z^d$ and $|x-y|=1$, then  we obtain
$$
( \big(\D_{\b(t)}- \D\big) f)_x=\tfrac{1}{2}
\!\!\!\!
\sum_{\be=(x,y)\in\cA}\!\!\!\!\big(1-e^{i\b(\be,t)}\big)f_y
=-i\sum_{\be=(x,y)\in\cA}
e^{i{\b(\be,t)\/2}} \sin {\b(\be,t)\/2}f_y.
$$
This and the estimate $|\b(\be,t)|\le w(t)F_y$ yield
$$
\begin{aligned}
|( \big(\D_{\b(t)}- \D\big)f)_x| \le w(t)\sum_{\be=(x,y)\in\cA}F_y |f_y|\le
 \sum_{\be=(x,y)\in\cA}|f_y|
\le \sqrt d \rt(\sum_{\be=(x,y)\in\cA} |F_yf_y|^2\rt)^{1\/2},
\end{aligned}
$$
since the degre ${\vk}_x= d$. Then the estimate \er{ii} implies
$$
\begin{aligned}
 \|(\D_{\b(t)}- \D\big) f\|^2\le
 w(t)^2d\sum_{x\in\Z^d} \sum_{|x-y|=1} |F_yf_y|^2
\le w(t)^2d^2 \|Ff\|^2,
\end{aligned}
$$
which together with Condition R  yields
\[\lb{R5}
\begin{aligned}
\|X(s) f\|\le \|(\D_{\b(s)}- \D\big) f\|+ \|v(s) f\|\le
 w(s)(d+1) \|Ff\|.
\end{aligned}
\]
From \er{R1} - \er{R5} we obtain
$$
\begin{aligned}
\|U(0,t)e^{-it\D}f-f\|\le \int_0^t \|X(s)e^{-is\D}f\|ds\le
(d+1)\int_0^t w(s)\|F e^{-is\D}f\|ds
\\
\le (d+1)\|w\|_{L^2(\R_+)}
\rt(\int_0^\iy \|b e^{-is\D}f\|^2ds \rt)^{1\/2}
\le (d+1)C_0\|w\|_{L^2(\R_+)}\|f\|,
\end{aligned}
$$
and similar arguments imply for a large time
$$
\begin{aligned}
\|W_+f-U(0,t)e^{-it\D}f\|\le
(d+1)\int_t^\iy w(s)\|F e^{-is\D}f\|ds
\\
\le (d+1)\int_t^\iy w(s)\|F e^{-is\D}f\|ds\le (d+1)C_o\|f\|
\rt(\int_t^\iy w(s)^2ds \rt)^{1\/2},
\end{aligned}
$$
which yields the norm convergence.
Thus the wave operators $W_\pm$ in \er{r}
 exists and unitary.

Let $q\ne 0$. In order to treat the potential $q$ we define a gauge
transformation $J(t)=e^{iQ(t)}, t\in \R$, where $Q_x(t)=\int_0^t q_x(s)ds$.
From  Condition R we have
$$
\1-J(t)=-2ie^{iQ(t)/2}\sin [iQ(t)/2],\qq |1-J_x(t)|\le Q_x^0,
$$
where $Q_x^0=\sup_{|t|\ge 1}|Q_x(t)|=o(1)$ as $|x|\to \iy$
and $Q^0$ is a compact operator in $\ell^2(\Z^d)$.
 Then from \er{R1} we obtain
\[
\begin{aligned}
\label{R3}
   W_+=s-\lim U(0,t) e^{-it\D}=s-\lim U(0,t)J(t) e^{-it\D}
   \qqq  as\qqq t\to +\iy.
\end{aligned}
\]
Then  we have the following:
\[
\lb{R4} U(0,t)J(t)e^{-it\D}f=f-i\int_0^t U(0,s)J(s)Y(s)e^{-is\D}fds,
\]
where
$$
Y(s)=J(s)^*(\D_{\b(s)}-\D+v) J(s)=J(s)^*(\D_{\b(s)}-\D) J(s)+v
=\D_{\vp(t)}+v,
$$
where the magnetic potential $\vp(\be,t)=\b(\be,t)+Q_x(t)-Q_y(t),
\qq \be=(x,y)$. The new magnetic potential $\vp(\be,t)$
and the potential $v$ satisfy Condition R.
Then by the above arguments, the wave operators $W_\pm$
 exist and are unitary.
\BBox


\end{document}